\documentclass[journal,onecolumn]{IEEEtran}

\usepackage{cite}
\ifCLASSINFOpdf
   \usepackage[pdftex]{graphicx}
\else
\fi
\ifCLASSOPTIONconference
\else
\fi
\usepackage[cmex10]{amsmath}
\usepackage{array}
\usepackage{mdwmath}
\usepackage[font=footnotesize]{subfig}
\usepackage{url}

\usepackage{amsmath,dsfont}
\usepackage{amssymb}
\usepackage{tikz}
\usepackage{caption}
\usepackage{algorithm}
\usepackage{algpseudocode}
\usepackage{pgfplots}
\usepackage{multirow}
\newtheorem{theorem}{Theorem}
\newtheorem{claim}{Claim}

\newtheorem{lemma}{Lemma}
\newtheorem{corollary}{Corollary}
\newtheorem{definition}{Definition}

\newtheorem{example}{Example}

\newcommand{\bproof}{ \begin{IEEEproof} }
\newcommand{\eproof}{ \end{IEEEproof} }
\newcommand{\beqno}{ \begin{equation*} }
\newcommand{\eeqno}{ \end{equation*} }
\newcommand{\beqa}{\begin{eqnarray*} }
\newcommand{\eeqa}{\end{eqnarray*} }
\newcommand{\beq}{ \begin{equation} }
\newcommand{\eeq}{ \end{equation} }

\newcommand{\calC}{\mathcal{C}}
\newcommand{\calD}{\mathcal{D}}

\newcommand{\calT}{\mathcal{T}}

\newcommand{\calU}{\mathcal{U}}

\newcommand{\dsW}{\mathds{W}}
\newcommand{\dsZ}{\mathds{Z}}
\newcommand{\dsD}{\mathds{D}}

\begin{document}
\title{Improved Lower Bounds for Coded Caching}
\author{\IEEEauthorblockN{Hooshang~Ghasemi and Aditya~Ramamoorthy\\}
\IEEEauthorblockA{Dept. of Electrical \& Computer Eng.\\ Iowa State University, Ames, IA 50011 \\
Email: \{ghasemi, adityar\}@iastate.edu}
\thanks{The material in this work has appeared in part at the 2015 IEEE International Symposium on Information Theory.}
}

\maketitle

\begin{abstract}
Content delivery networks often employ caching to reduce transmission rates from the central server to the end users. Recently, the technique of coded caching was introduced whereby coding in the caches and coded transmission signals from the central server are considered. Prior results in this area demonstrate that carefully designing the placement of content in the caches and designing appropriate coded delivery signals from the server allow for a system where the delivery rates can be significantly smaller than conventional schemes. However, matching upper and lower bounds on the transmission rate have not yet been obtained.
In this work, we derive tighter lower bounds on the coded caching rate than were known previously. We demonstrate that this problem can equivalently be posed as a combinatorial problem of optimally labeling the leaves of a directed tree. Our proposed labeling algorithm allows for significantly improved lower bounds on the coded caching rate. Furthermore, we study certain structural properties of our algorithm that allow us to analytically quantify improvements on the rate lower bound for general values of the problem parameters. This allows us to obtain a multiplicative gap of at most four between the achievable rate and our lower bound.

%

\end{abstract}

\IEEEpeerreviewmaketitle
\section{Introduction}

Content distribution over the Internet is an important problem and is the core business of several enterprises such as Youtube, Netflix, Hulu etc. The operation of such large scale systems presents several challenges, including (but not limited to) storage of the data, ensuring reliable availability and efficient content delivery. One commonly used technique to facilitate delivery is content caching \cite{wessels}. The main idea in ``conventional content caching" is to store relatively popular content in local memory either on the desired device or in a device at the edge of the network such as an intermediate router. This local memory is referred to as the cache. Upon request, this cached content is used to serve the clients, thus reducing the number of bits transmitted from the server and thereby reducing overall network congestion. Note that even web browsers, routinely cache the content of popular websites on a local machine to speed up the loading of webpages. 

Historically, content caching algorithms have attempted to optimize the placement of content in the caches so that the average number of bits that are transmitted from the central server to the end users is minimized \cite{MeyersonMP01,KorupoluPR99,BorstGW10,tanM13}. This often requires some knowledge on the popularity of file requests \cite{wolman99,breslau_et_al99,applegate_et_al10} made by the users. Moreover, the typical approach is to cache a certain fraction of the file and to obtain the remaining parts from the server when the need arises. Coding in the content of the cache and/or coding in the transmission from the server are typically not considered.

The work of \cite{maddahN14} introduced the problem of coded caching, where there is a server with $N$ files and $K$ users each with a cache of size $M$. The users are connected to the server by a shared link (see Fig. \ref{fig:block_diag}). In each time slot each user requests one of the $N$ files.  There are two distinct phases in coded caching.
\begin{itemize}
\item {\it Placement phase:} In this phase, the content of caches is populated. This phase should not depend on the actual user requests (which are assumed to be arbitrary). Typically, this placement phase can be executed in the {\it off-peak} hours where the amount of network traffic is low.

\item {\it Delivery phase:} In this phase, each of the $K$ users request one of the $N$ files. The server transmits a signal of rate $R$ over the shared link that simultaneously serves to satisfy the demands of each of the users.
\end{itemize}
The work of \cite{maddahN14} demonstrates that a carefully designed placement scheme and a corresponding delivery scheme achieves a rate that is significantly lower than conventional caching.  While coded caching promises very significant gains in transmission rates, at this point we do not have matching upper and lower bounds on the $(R,M)$ pairs for a given $N$ and $K$.
\definecolor{myblue}{RGB}{53,122,183}
\begin{figure}[t]
\centering
\includegraphics[scale=1]{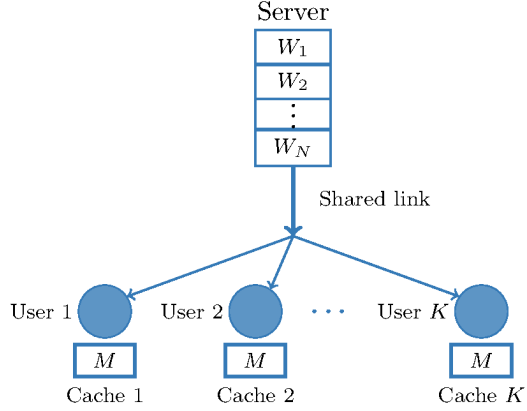}
\caption{Block diagram of coded caching system.}
\label{fig:block_diag}
\end{figure}

In this work our main contribution is in developing improved lower bounds on the required rate for the coded caching problem. We demonstrate that the computation of this lower bound can be posed as a combinatorial labeling problem on a directed tree. In particular, our method generates lower bounds on $\alpha R + \beta M$, where $\alpha, \beta$ are positive integers. We demonstrate that a careful analysis of the underlying combinatorial structure of the problem allows us to obtain significantly better lower bounds than those obtained in prior work \cite{maddahN14,sengupta2015improved,ajaykrishnan2015critical}. In addition, our machinery allows us to show that the achievable rate of \cite{maddahN14} is within a multiplicative factor of four of our proposed lower bound.

This paper is organized as follows. Section \ref{sec:background} discusses the background, related work and summarizes the main contributions of our work. Section \ref{sec:lower_bound} presents our proposed lower bound technique. The multiplicative gap between the achievable rate and our lower bound is outlined in Section \ref{sec:mult_gap}. Our proposed strategy also applies to certain variants of the coded caching problem that have been discussed in the literature; this is explained in Section \ref{sec:variants}. There have been some other approaches presented in the literature \cite{maddahN14,sengupta2015improved,ajaykrishnan2015critical} for improving the lower bound on the coded caching rate. We present comparisons between our approach and the other approaches in Section \ref{sec:comparison}. We conclude the paper with a discussion of opportunities for future work in Section \ref{sec:conclusions}.
\section{Background, Related Work and Summary of Contributions}
\label{sec:background}

In a coded caching system there is a server that contains $N$ files, denoted $W_i, i =1, \dots, N$, each of size $F$ bits. There are $K$ users that are connected to the central server by means of a shared link. Each user has a local cache memory of size $MF$ bits; we denote the cache content by the symbol $Z_i$ (which is a function of $W_1, \dots, W_N$). In each time slot, the $i$-th user demands the file $W_{d_i}$ where $d_i \in \{1, \dots, N\}$.
The coded caching problem has two distinct phases. In the {\it placement phase}, the content of caches is populated; this phase should not depend on the actual user requests (which are assumed to be arbitrary). In the {\it delivery phase}, the server transmits a potentially coded signal that serves to satisfy the demands of each of the users. A pair $(M,R)$ is said to be achievable if for every possible request pattern (there are $N^K$ of them), every user can recover its desired file with high probability for large enough $F$. We let $R^\star(M)$ denote the infimum of all such achievable rates for a given $M$.

The coded caching problem can be formally described as follows.
Let $[m] = \{1, \dots, m\}$, where $m$ is a positive integer.
Let $\{W_n\}_{n=1}^N$ denote $N$ independent random variables (representing the files) each uniformly distributed over $[2^F]$. The $i$-th user requests the file $W_{d_i}$, where $d_i \in [N]$. A $(M,R)$ system consists of the following.
\begin{itemize}
\item $K$ caching functions, $Z_i \triangleq \phi_i(W_1, \dots, W_N)$ where $\phi_i: [2^F] \rightarrow [2^{\lfloor FM \rfloor}]$.
\item A total of $N^K$ encoding functions $\varphi_{d_1, \dots, d_K}(W_1, \dots, W_N)$, so that the delivery phase signal $X_{d_1, \dots, d_K} \triangleq \varphi_{d_1, \dots, d_K}(W_1, \dots, W_N)$. Here, $\varphi_{d_1, \dots, d_K}:[2^F]^N \rightarrow [2^{\lfloor FR \rfloor}]$.
\item For each delivery phase signal and each user, we define appropriate decoding functions. There are a total of $K N^K$ of them.  For the $k$-th user $\mu_{d_1, \dots, d_K;k} (X_{d_1, \dots, d_K}, Z_k)$, $k=1,\dots,K$ so that decoded file $\hat{W}_{d_1, \dots, d_K;k} \triangleq \mu_{d_1, \dots, d_K;k} (X_{d_1, \dots, d_K}, Z_k)$. Here  $\mu_{d_1, \dots, d_K;k}:[2^{\lfloor RF \rfloor}] \times [2^{\lfloor FM \rfloor}] \rightarrow [2^F]$.
\end{itemize}
The probability of error is defined as
\begin{align*}
\max_{(d_1, \dots, d_K) \in [N]^K} \max_{k \in [K]} P(\hat{W}_{d_1, \dots, d_K; k} \neq W_{d_k}).
\end{align*}
\begin{definition}
The pair $(M,R)$ is said to be achievable if for $\epsilon > 0$, there exists a file size $F$ large enough so that there exists a $(M,R)$ caching scheme with probability of error at most $\epsilon$. We define
\begin{align*}
R^{\star}(M) = \inf \{R: (M,R) \text{~is achievable}\}.
\end{align*}
\end{definition}
In this setting, it is not too hard to see that the best that a conventional caching system can do is to simply store an $M/N$ fraction of each file in each of the caches. In order to satisfy the demands of the user, the server has to transmit the remaining $(1-M/N)$ fraction of each of the $K$ files. Thus the transmission rate (normalized by $F$) is given by
\begin{align}
\label{eq:uncoded_rate}
R_{U}(M) = \min(N,K) \bigg{(}1 -\frac{M}{N}\bigg{)}.
\end{align}
Note that $\min(N,K)$ is the transmission rate in the absence of any caching. In \cite{maddahN14}, the factor $(1-M/N)$ is referred to as the {\it local caching gain} as it is gain that is obtained purely from the cache, without any optimization of the transmission from the server. In the setting where we perform nontrivial coding in the cache and delivery phase encoding functions, \cite{maddahN14} demonstrates that a carefully designed placement scheme and a corresponding delivery scheme achieves a rate
\begin{align}
\label{eq:coded_rate}
R_{C}(M) = K\bigg{(}1 -\frac{M}{N}\bigg{)} \cdot \min \bigg{\{}\frac{1}{1 + KM/N}, \frac{N}{K}\bigg{\}},
\end{align}
where $M \in \{0,N/K, 2N/K,\dots, N\}$. Other values of $M$ are obtained by time-sharing between different solutions.

The factor $\frac{1}{1 + KM/N}$ which definitely dominates when $N \geq K$ is referred to as the {\it global caching gain}. It is to be noted that global caching gain depends on the overall cache size across all the users (owing to the term $KM/N$ in the denominator) whereas the local caching gain only depends on the per-user cache size (owing to the term $1 - M/N$). The work of \cite{maddahN14} also shows that the rate $R_C(M)$ is within a factor of $12$ of the information theoretic optimum for all values of $N, K$ and $M$. Furthermore, they compare their achievable rate ({\it cf.} eq. (\ref{eq:coded_rate})) to a cutset bound that can be expressed as follows.
\begin{align}
\label{eq:cutset_bd}
R^\star(M) \geq \max_{s \in \{1, \dots, \min(N,K)\}} \bigg{(} s - \frac{s}{\lfloor N/s \rfloor} M \bigg{)}.
\end{align}

\subsection{Related work}
Coded caching is related to but different from the index coding problem \cite{barBJK11}. In the index coding problem, there are $N'$ sources such that $i$-th source has message $W_i$, $i =1, \dots, N'$. There are $K$ terminals, each of which has some subset of $\{W_1,\dots,W_{N'}\}$ available. In addition, each terminal requests a certain subset of the messages $\{W_1,\dots,W_{N'}\}$. The aim in the index coding problem is to minimize the number of bits that are transmitted on the shared link so that the demands of each user are satisfied. It is well recognized that the index coding problem for arbitrary side information is a computationally hard problem where nonlinear codes may be necessary \cite{barBJK11,lubetzkyS09}. In particular, the optimal {\it linear} index code corresponds to minimizing the rank of an appropriately defined matrix over a finite field. This so called minrank problem \cite{barBJK11} is also known to be computationally hard. It can be observed that for a fixed but {\it uncoded} cache content and a fixed set of demands of the various users, the problem of determining the optimal delivery phase signal in the coded caching problem is equivalent to an index coding problem. Note however, that in the coded caching problem, we allow the cache content to be coded.

Since the original work of \cite{maddahN14}, there have been several aspects of coded caching that have been investigated. Reference \cite{maddahN14mr_tradeoff} considers the scenario of decentralized caching when the placement phase is driven by the users who randomly populate their caches with subsets of the files stored at the server.
Approaches for updating the cache content are considered in \cite{pedarsaniMN14} and the case of files with different popularity scores are considered in \cite{maddahN14nonuniform_demand} and \cite{jiTLC14zipf,hachemKD14a}. Security issues in this domain are considered in \cite{senguptaTC13}. The work of \cite{karamNMD14} considers the more general case of hierarchical coded caching, where certain intermediate nodes in the network are equipped with potentially larger caches and investigates methods for minimizing the overall traffic in such networks (see also \cite{hachemKD14}). Coded caching where each user requests multiple files was investigated in \cite{jiTLC14}. The case of device-to-device (D2D) wireless networks where there is no central server was examined in \cite{jiCM13,senguptabeyondd2d}. Systems with files of differing sizes were examined in \cite{zhang2015coded}.

In addition to these contributions, there have been other lines of work that deal with content caching. In a parallel line of work \cite{shanmugamGDMC13,golrezaeiMDC13,jiCM13,yueYCGZ14} consider the problem of femtocaching in a wireless setting where in addition to a central server (or base station), there are helpers (with caches) interspersed in a cell that help the end users satisfy their demands. The goal is again to consider caching strategies that minimize the overall rate, but the solution approaches do not consider the worst case rate over all possible demand patterns; instead the popularity scores of the different files are explicitly taken into account. Moreover, while coding is considered, it is conceptually different in the sense that the coding is only restricted to parts of the same file and coding across different files is not considered.

There has also been parallel work on establishing lower bounds for the coded caching problem. In \cite{sengupta2015improved}, the Han's inequality was leveraged to obtain an improved lower bound. A multiplicative gap of $8$ between their lower bound and the achievable rate in eq. (\ref{eq:coded_rate}) ws established. The work of \cite{ajaykrishnan2015critical} also presents a lower bound technique. As discussed in Section \ref{sec:comparison}, their technique can be considered as a special case of our work. The specific case of $N=K=3$ was considered in \cite{tian2015note} via a computational approach. We compare our technique with these other approaches in Section \ref{sec:comparison}.

\subsection{Summary of our contributions}
In this work our main contribution is in developing improved lower bounds on the rate for the coded caching problem. We show that the cutset based bound in eq. (\ref{eq:cutset_bd}) is significantly loose and propose a larger class of lower bounds that are significantly tighter.
Our specific contributions include the following.
\begin{itemize}
\item We demonstrate that the computation of our lower bound can be posed as a combinatorial labeling problem on a directed tree. Our method generates lower bounds on $\alpha R^\star + \beta M$, where $\alpha, \beta$ are positive integers. While the cutset bound only optimizes over at most $\min(N,K)$ choices, our technique allows us to consider many more $(\alpha, \beta)$ pairs\footnote{The cutset bound can be considered as a special case of our bound}.
\item  We demonstrate that a careful analysis of the underlying combinatorial structure of the problem allows us to obtain significantly better lower bounds than those obtained in prior work. For a given pair $(\alpha, \beta)$ and number of users $K$, it is intuitively clear that the lower bound on $\alpha R^\star + \beta M$ will be large if the number of files $N$ is large. We define the notion of a saturated instance, which are directed trees and corresponding labelings that give the largest possible lower bound (using our technique) using as few files as possible. An analysis of saturated instances allows us to always improve on the cutset bound and in most ranges of $M$, our bound is strictly better.
\item Our machinery allows us to show that the achievable rate of \cite{maddahN14} is within a multiplicative factor of four of our proposed lower bound for all values of $N$ and $K$. This is possible by analyzing some combinatorial properties of saturated instances. Note that the multiplicative gap of four is currently the best known for this problem.
\item Our proposed technique also applies to other variants of coded caching problem. We discuss the application of our work to the case of D2D wireless networks and coded caching with multiple requests as well.
\end{itemize}

\begin{figure}[!t]
\centering
\includegraphics[scale=0.5]{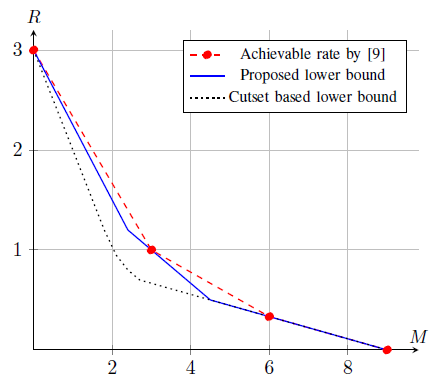}
\caption{{\small An example of a coded caching system with $N=9$ files, $K=3$ users. Note that the proposed lower bound is better than the cutset bound and matches the achievable rate points at multiples of $N/K$.}}
\label{Fig:showsuperioteryofresult}
\end{figure}
As an example, Fig. \ref{Fig:showsuperioteryofresult} illustrates the tightness of the proposed lower bound for a coded caching system with a server that contains $N=9$ files and $K=3$ users. Specifically, our proposed bound demonstrates the optimality of the achievable scheme for values of $M$ that are integer multiples of $N/K$ in this specific case. 


\section{Lower Bound on $R^{\star}(M)$}
\label{sec:lower_bound}
In this section we present our proposed lower bound on $R^\star(M)$. We begin with an example that demonstrates the core idea of our approach.
\begin{example}
\label{exp:TreeInstanceExmp}
Consider a coded caching system with $N = K = 3$. Then, the following sequence of information theoretic inequalities hold.
\begin{align*}
&2R^{\star}F + 2MF \geq H(Z_1, X_{123}) + H(Z_2, X_{312}) \\
&\stackrel{(a)}{=} I(W_1;Z_1, X_{123}) +  H(Z_1, X_{123} | W_1) + I(W_1;Z_2, X_{312}) + H(Z_2, X_{312} | W_1)\\
&= H(W_1) - H(W_1 | Z_1, X_{123}) + H(Z_1, X_{123} | W_1) + H(W_1) - H(W_1|Z_2, X_{312}) + H(Z_2, X_{312} | W_1)\\
&\stackrel{(b)}{\geq } F(1 - \epsilon) + F(1 - \epsilon)+ H(Z_1, Z_2, X_{123}, X_{312} | W_1)\\
&= 2F(1 -\epsilon) + I(W_2, W_3 ; Z_1, Z_2, X_{123}, X_{312} | W_1) + H(Z_1, Z_2, X_{123}, X_{312} | W_1, W_2, W_3)\\
&\stackrel{(c)}{\geq} 2F(1-\epsilon) + 2F(1-\epsilon) = 4F(1-\epsilon),
\end{align*}
where equality (a) holds by the definition of mutual information. Inequality (b) holds by Fano's inequality since the file $W_1$ can be recovered with $\epsilon$-error from the pairs $(Z_1, X_{123})$ and $(Z_2, X_{312})$ and by the fact that  conditioning reduces entropy. Similarly, inequality (c) holds by Fano's inequality since the files $W_2$ and $W_3$ can be recovered with $\epsilon$-error from $(Z_1, Z_2, X_{123}, X_{312})$. This holds for arbitrary $\epsilon > 0$ and $F$ large enough. Dividing throughout by $F$ we have the required result.
\end{example}

Thus, the key idea of the above bound is to choose the delivery phase signals in such a manner so that the various terms that are combined allow the ``reuse" of the same file multiple times. For instance, in step (a) of the above bound, we use the definition of mutual information to rewrite the terms $H(Z_1,X_{123})$ and $H(Z_2,X_{312})$. Note that both pairs $(Z_1, X_{123})$ and $(Z_2, X_{312})$ allow the recovery of the {\it same} file $W_1$, resulting in a contribution of $2F$ to the lower bound. On the other hand, the files $W_2$ and $W_3$ are recovered only once. The overall result is a lower bound of $4F$.

Thus, our lower bound works with judiciously chosen labels for the delivery phase signals and combines them with the cache signals in an appropriate way such that a given file is recovered a large number of times. It turns out that doing this systematically and tractably requires the development of several new ideas. For instance, the aforementioned chain of inequalities can be equivalently represented in terms of a directed tree with appropriate labels on its leaves and edges as shown in Fig. \ref{Fig:TreeInstanceExmp}. In particular, the leaves of the tree are labeled with cache signals $Z_1$ and $Z_2$ and delivery phase signals $X_{123}$ and $X_{312}$.
Each internal node of the tree corresponds to the operation of combining the signals and its outgoing edge is labeled by the newly recovered file(s), e.g., at node $u_1$, the file $W_1$ is recovered. Likewise at node $u^*$, the files $W_2$ and $W_3$ are recovered. The lower bound can be obtained by summing the cardinalities of the edge labels. Towards the goal of generating these bounds in a systematic manner, we introduce the following definitions.
\begin{definition} {\it Directed in-tree.}
A directed graph $\calT = (V, A)$, is called a directed in-tree if there is one designated node called the root such that from any other vertex $v \in V$ there is exactly one directed path from $v$ to the root.
\end{definition}
The nodes in a directed in-tree that do not have any incoming edges are referred to as the leaves. The remaining nodes, excluding the leaves and the root are called internal nodes. Each node in a directed in-tree has at most one outgoing edge. We have the following definitions for a node $v \in V$.
\begin{align*}
out(v) &= \{u \in V : (v,u) \in A\}, \text{~(outgoing neighbor) and},\\
in(v) &= \{u \in V: (u,v) \in A\} \text{~(incoming neighbor set).}\\
in-edge(v) &= \{e \in A: e = (u,v)\} \text{~(incoming edge set)}.
\end{align*}
In this work, we exclusively work with trees which are such that the in-degree of the root equals 1. There is a natural topological order in $\calT$ whereby for nodes $u \in \calT$ and $v \in \calT$, we say that $u \succ v$ if there exists a sequence of edges that can be traversed to reach $v$ from $u$. This sequence of edges is denoted $path(u,v)$.
\begin{algorithm}[!t]
\caption{Lower Bound Algorithm}\label{Alg:Labeling}
\algrenewcommand\algorithmicrequire{\textbf{Input:}}
\algrenewcommand\algorithmicensure{\textbf{Output:}}
\algrenewcommand\algorithmicfunction{\textbf{Initialization}}
\begin{algorithmic}[1]
\Require $\calT = (V, A)$ with leaves $v_1, \ldots, v_\ell$ and $\{ label(v_i)\}_{i=1}^\ell$, such that $\dsW(v_i) = \emptyset, i=1,\dots, \ell$.
\algrenewcommand\algorithmicrequire{\textbf{Initialization:}}
\Require
\For{$i \gets 1, \dots \ell$}
\State $W_{new}(v_i) = \Delta(v_i,v_i)$.
\State $x_{(v_i,out(v_i))} = W_{new}(v_i)$.
\State $y_{(v_i,out(v_i))} = |W_{new}(v_i)|$.
\EndFor
\While{there exists an unlabeled edge}
\State Pick an unlabeled node $u \in V$ such that all edges in $in-edge(u)$ are labeled.
\State $\dsW(u) = \cup_{v \in in(u)} \dsW(v) \cup W_{new}(v)$.
\State $\dsZ(u) = \cup_{v \in in(u)} \dsZ(v)$.
\State $\dsD(u) = \cup_{v \in in(u)} \dsD(v)$.
\State $W_{new}(u) = \Delta(u,u) \setminus \dsW(u)$.
\State $x_{(u,out(u))} = W_{new}(u)$.
\State $y_{(u,out(u))} = |W_{new}(u)|$.
\EndWhile
\Ensure $L=\sum_{e \in A} y_{e} $.
\end{algorithmic}
\end{algorithm}

\begin{definition}
\label{defn:meeting_point}
{\it Meeting point of nodes in a directed tree.} Consider nodes $v_1$ and $v_2$ in a directed in-tree $\calT = (V,A)$. We say that $v_1$ and $v_2$ meet at node $u$ if there exist $path(v_1,u)$ and $path(v_2,u)$ in $\calT$ such that $path(v_1,u) \cap path(v_2,u) = \emptyset$. As there exists a path from any node in $\calT$ to the root node, it follows that the existence of node $u$ is guaranteed.
\end{definition}
\par Let $D = \cup_{d_1 \in [N], \dots, d_K \in [N]} \{ X_{d_1, \dots, d_K}\}$.
\begin{definition} {\it Labeling of directed in-tree.} Each node $v \in \calT$ is assigned a label, denoted $label(v)$, which is a subset of $\{W_1, \dots, W_N\} \cup \{Z_1, \dots, Z_K\} \cup D$. Moreover, we also specify $\dsW(v) \subseteq \{W_1, \dots, W_N\} $, $\dsZ(v) \subseteq \{Z_1, \dots, Z_K\} $ and $\dsD(v) \subseteq D$ so that $label(v) =  \dsW(v) \cup \dsZ(v) \cup \dsD(v)$. 
\end{definition}
In our formulation, the leaf nodes are denoted $v_i, i = 1, \dots, \ell$ are such that $\dsW(v_i) = \emptyset$.
\begin{definition}
We say that a singleton source subset $\{W_i\}$ is recoverable from the pair $(Z_j, X_{d_1, \dots, d_K})$ if $d_j = i$. 
Similarly, for a given set of caches $Z' \subseteq \{Z_1, \dots, Z_K\}$ and delivery phase signals $D' \subseteq D$, we define a set $Rec(Z',D') \subseteq \{W_1, \dots, W_N\}$ to be the subset of the sources that can be recovered from pairs of the form $(Z_i, X_{J})$ where $Z_i \in Z'$ and $J$ is a multiset of cardinality $K$ with entries from $[N]$ such that $X_{J} \in D'$.
\end{definition}
We let the entropy of a set of random variables equal the joint entropy of all the random variables in the set. We also let $[x]^+ = \max(x,0)$.

Given a directed tree $\calT$ with appropriate labels on its leaves we present an algorithm that generates an inequality of the form $\alpha R^{\star} + \beta M \geq L(\alpha, \beta)$. For nodes $u,v \in \calT$, we define the following. 
\begin{align}
\Delta(u,v) &= Rec(\dsZ(u),\dsD(v)), \text{~and} \nonumber\\
W_{new}(u) &= \Delta(u,u) \setminus \dsW(u). \label{eq:wnew_def}
\end{align}

\begin{figure}[!t]
\centering
%
%
%
\includegraphics[scale=0.5]{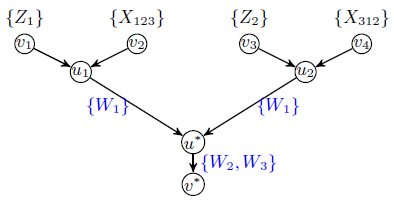}
\caption{{\small Problem instance for Example \ref{exp:TreeInstanceExmp}. For clarity of presentation, only the $W_{new}(u)$ label has been shown on the edges.}}
\label{Fig:TreeInstanceExmp}
\end{figure}



Algorithm \ref{Alg:Labeling} operates as follows. It takes as input a directed in-tree $\calT$ where each leaf $v_i, i = 1, \dots, \ell$ has labels $\dsZ(v_i)$ and $\dsD(v_i)$ ($\dsW(v_i)$ is set to $\emptyset$). The algorithm determines the files that are recovered at each $v_i$ and labels the corresponding outgoing edge with $W_{new}(v_i)$ and $|W_{new}(v_i)|$. Following this, the algorithm propagates the labels further down the tree in the following manner. For a given node $u$ whose incoming edges are labeled, we set  $\dsZ(u) = \cup_{v \in in(u)} \dsZ(v)$ and $\dsD(u) = \cup_{v \in in(u)} \dsD(v)$, i.e., each of these labels is set to the union of the corresponding labels of the nodes that belong to the incoming node set of $u$. Next, it sets $\dsW(u) = \cup_{v \in in(u)} \dsW(v) \cup W_{new}(v)$, i.e., in addition to the $\dsW$-labels of the incoming node set, $\dsW(u)$ also contains the new files that are recovered on the incident edges. Note that at each internal node certain cache signals and delivery phase signals {\it meet}, e.g., $Z_1$ and $X_{123}$ meet at node $u_1$ in Fig. \ref{Fig:TreeInstanceExmp}. The outgoing edge of an internal node is labeled by the {\it new} files that are recovered at the node, e.g., at $u_1$ the signals $Z_1$ and $X_{123}$ recover the file $W_1$. We call a file new if it has not been recovered upstream of a given node. In a similar manner at $u^*$ one can recover all the files $W_1, \dots, W_3$; however only the set $\{W_2,W_3\}$ is labeled on edge $(u^*, v^*)$ as $W_1$ was recovered upstream. This process is continued recursively, i.e., we label the outgoing edges with the new files that are recovered at node $u$, propagate the labels and continue thereafter. The algorithm continues until it labels the last outgoing edge.

It can be seen that the operation of Algorithm \ref{Alg:Labeling} is in one to one correspondence with the new files recovered in the sequence of inequalities in the lower bound. For example, the outgoing labels of $u_1$ and $u_2$ in Fig. \ref{Fig:TreeInstanceExmp} correspond to step (a) in the inequalities in Example \ref{exp:TreeInstanceExmp}.
We formalize this statement in the Appendix (Lemma \ref{lemma:general_lower_bd}) where we show that a valid lower bound is always obtained when applying Algorithm \ref{Alg:Labeling}.


\begin{definition} {\it Problem Instance.}
Consider a given tree $\calT$ with leaves $v_i, i = 1, \dots, \ell$ that are labeled as discussed above. Let $\alpha = \sum_{i=1}^{\ell} |\dsD(v_i)|$ and $\beta = \sum_{i=1}^{\ell} |\dsZ(v_i)|$. Suppose that the lower bound computed by Algorithm \ref{Alg:Labeling} equals $L$. We define the associated problem instance as $P(\calT, \alpha, \beta, L, N, K)$. We also define $\hat{\alpha} = |\cup_{i=1}^{\ell} \dsD(v_i)|$ and $\hat{\beta} = |\cup_{i=1}^{\ell} \dsZ(v_i)|$. A problem instance $P(\calT, \alpha, \beta, L, N, K)$ is said to be optimal if all instances of the form $P'(\calT',\alpha, \beta, L', N, K)$ are such that $L' \leq L$.
\end{definition}
It is worth emphasizing that $\hat{\alpha} \leq \alpha$ and $\hat{\beta} \leq \beta$ as some cache and delivery phase signals may be repeated.

In the subsequent discussion, we focus on understanding the characteristics of optimal problem instances. Towards this end, we shall often start with a problem instance $P$ and modify it in appropriate ways to arrive at another instance $P'$. For ease of presentation, when needed we shall refer to quantities in instance $P$($P'$) by using the corresponding superscripts. For example for a node $u$ in $P$ ($P'$), we will denote the set of new files by $W_{new}^P(u)$ ($W_{new}^{P'}(u)$).

It is not too hard to see that it suffices to consider directed trees whose internal nodes have an in-degree at least two. In particular, if $u$ has in-degree equal to $1$, it is evident that $W_{new}(u) = \emptyset$ and thus, $|W_{new}(u)| = 0$. In addition, we claim that w.l.o.g. it suffices to consider trees where internal nodes have in-degree at most two.  Therefore, we will assume that all internal nodes have degree equal to two. More specifically, we can show the following property of problem instances (the proof appears in the Appendix).
\begin{claim}
\label{claim:incomingedgelimits}
Consider a problem instance $P(\calT, \alpha, \beta, L, N, K)$ such that there exists a node $u \in \calT$ with $|in(u)| \geq 3$. Then, there exists another instance $P'(\calT', \alpha, \beta, L', N, K)$ where $L' \geq L$ and $|in(u)| \leq 2$ for all nodes $u \in \calT'$.
\end{claim}
Henceforth, we assume that all internal nodes in the problem instances under consideration have in-degree equal to two. Claim \ref{claim:incomingedgelimits} can also be used to conclude that each leaf $v$ in an instance $P$ is such that either $|\dsZ(v)| = 1$ or $|\dsD(v)| = 1$ but not both. Indeed, if there exists a leaf $v$ that violates this condition, we can use the modification in the proof of Claim \ref{claim:incomingedgelimits} to replace $v$ by a directed in-tree so that the condition is satisfied.
If $|\dsZ(v)| = 1$, we call $v$ a cache node; if $|\dsD(v)| = 1$ we call it a delivery phase node. In the subsequent discussion we will assume that the delivery phase nodes are labeled in an arbitrary order $v_1, \dots, v_\alpha$ and the cache nodes from $v_{\alpha+1}, \dots, v_{\alpha+\beta}$, where we note that $\alpha+\beta = \ell$. Moreover, we let $\calD = \{v_1, \dots, v_\alpha\}$ and $\calC = \{v_{\alpha+1}, \dots, v_{\alpha+\beta}\}$.

\begin{figure}[t]
\centering
\includegraphics[scale=0.5]{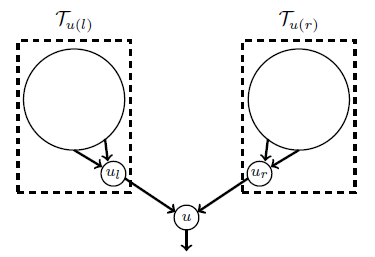}
\caption{For a given node $u \in \calT$, its in-neighbors are denoted $u_l$ and $u_r$. The corresponding subtrees are denoted $\calT_{u(l)}$ and $\calT_{u(r)}$ and are shown enclosed in the dotted boxes.}
\label{fig:subproblem_instances}
\end{figure}

In the tree $\calT$ corresponding to problem instance $P(\calT, \alpha, \beta, L, N , K)$, consider an internal node $u$ and the edge $e = (u,v)$. In the subsequent discussion, we shall use $\calT_u$ to refer to the subtree that has its last edge as $(u,out(u))$, i.e., the subtree that is rooted at $out(u)$. The incoming edges into $u$, denoted $(u_l, u)$ and $(u_r,u)$ are the last edges of the disjoint left and right subtrees denoted $\calT_{u(l)}$ and $\calT_{u(r)}$ respectively (see Fig. \ref{fig:subproblem_instances}).
Each of these subtrees defines a problem instance $P_l = P(\calT_{u(l)}, \alpha_l, \beta_l, L_l, N, K)$ and $P_r = P(\calT_{u(r)}, \alpha_r, \beta_r, L_r, N, K)$. We denote the set of delivery phase nodes and cache nodes in $\calT_{u(r)}$ by
\begin{align*}
\calD_{u(r)} &= \{v \in \calD: v \in \calT_{u(r)}\} \text{~and}\\
\calC_{u(r)} &= \{v \in \calC: v \in \calT_{u(r)}\},
\end{align*}
with similar definitions for $\calD_{u(l)}$ and $\calC_{u(l)}$. We also let
\begin{align*}
\calD_u &= \calD_{u(l)} \cup \calD_{u(r)}, \text{~and}\\
\calC_u &= \calC_{u(l)} \cup \calC_{u(r)}.
\end{align*}

Let $\Gamma_l = \cup_{v \in \calT_{u(l)}} W_{new}(v)$ and $\Gamma_r = \cup_{v \in \calT_{u(r)}} W_{new}(v)$, i.e., $\Gamma_l$ and $\Gamma_r$ are the subsets of $\{W_1, \dots, W_N\}$ that are used up in the problem instances $P_l$ and $P_r$ respectively. It can be observed that $\Gamma_l = \Delta(u_l, u_l)$ and $\Gamma_r = \Delta(u_r, u_r)$.

We shall often need to reason about the files recovered at the node $u$ from the different subtrees. For instance, the set of cache nodes in $\calT_{u(r)}$ and the delivery phase signals in $\calT_{u(l)}$ meet and recover a subset of the files at $u$. This set of files corresponds to those recovered from $\dsZ(u_r) \setminus \dsZ(u_l)$ and $\dsD(u_l)$, and can be informally thought of as the {\it files recovered when going from right to left}. Accordingly, we have the following definitions.
\begin{align*}
\Delta_{rl}(u) &= Rec(\dsZ(u_r) \setminus \dsZ(u_l),  \dsD(u_l)), \text{~and}\\
\Delta_{lr}(u) &= Rec(\dsZ(u_l) \setminus \dsZ(u_r),  \dsD(u_r)).
\end{align*}
Note that by definition, we have
\begin{align*}
\Delta(u,u) &= Rec(\dsZ(u),\dsD(u))\\
&=Rec(\dsZ(u_l) \cup \dsZ(u_r), \dsD(u_l) \cup \dsD(u_r))\\
&=Rec(\dsZ(u_l), \dsD(u_l)) \cup Rec(\dsZ(u_r), \dsD(u_r)) \cup Rec(\dsZ(u_l), \dsD(u_r)) \cup Rec(\dsZ(u_r), \dsD(u_l))\\
&\stackrel{(a)}{=} Rec(\dsZ(u_l), \dsD(u_l)) \cup Rec(\dsZ(u_r), \dsD(u_r)) \cup Rec(\dsZ(u_l) \setminus \dsZ(u_r), \dsD(u_r)) \cup Rec(\dsZ(u_r) \setminus \dsZ(u_l), \dsD(u_l))\\
&= \underbrace{\Delta(\dsZ(u_l),\dsD(u_l))}_{\text{from $\calT_{u(l)}$}} \cup \underbrace{\Delta(\dsZ(u_r),\dsD(u_r))}_{\text{from $\calT_{u(r)}$}} \cup \Delta_{lr}(u) \cup \Delta_{rl}(u), \text{~and}\\
\dsW(u) &= \Delta(\dsZ(u_l),\dsD(u_l)) \cup \Delta(\dsZ(u_r),\dsD(u_r)),
\end{align*}
where $(a)$ follows since the $Rec(\dsZ(u_l), \dsD(u_r))$ potentially contains some files that have already been recovered in $Rec(\dsZ(u_r), \dsD(u_r))$. The other equality holds because of similar reasoning.
Therefore, it follows that
\begin{align}
W_{new}(u) &= \Delta(u,u) \setminus \dsW(u) \nonumber\\
&= \Delta_{rl}(u) \cup \Delta_{lr}(u) \setminus \dsW(u). \label{eq:wnew_2}
\end{align}

Note that based on Algorithm \ref{Alg:Labeling}, we can conclude that
\begin{align}
\label{eq:W-lab-union}
\dsW(u) &= \cup_{v\in \{u_r,u_l \}} \dsW(v) \cup W_{new}(v) \nonumber\\
&= \cup_{v \succ u} W_{new}(v) \text{~(by arguing inductively)}.
\end{align}


\begin{algorithm}[!t]
\caption{Computing $\psi$}\label{Alg:AssignPsi}
\algrenewcommand\algorithmicrequire{\textbf{Input:}}
\algrenewcommand\algorithmicfunction{\textbf{Initialization}}
\begin{algorithmic}[1]
\Require $P(\calT,\alpha, \beta, L, N, K)$, Array $\Omega(u, \delta_u)$, where $u \in \calT$, $\delta_u \subseteq W_{new}(u), |\delta_u| = 1$.
\Function{}{}
\ForAll {$u \in \calT$, $\delta_u \subseteq W_{new}(u)$ where $|\delta_u| = 1$}
\State $\Omega(u, \delta_u) \gets 0$,
\EndFor
\EndFunction
\For{$i \leftarrow 1$ to $\alpha$}
\ForAll{$v' \in \calC$}
\State Let $u$ be the meeting point of $v_i$ and $v'$.
\State $\delta_u = \Delta(v', v_i)$.
\If{$\delta_u \in W_{new}(u)$ and $\Omega(u,\delta_u)==0$}
\State $\psi(v_i, v') \gets 1$, and $\Omega(u, \delta_u) \gets 1$.
\Else
\State $\psi(v_i, v') \gets 0$.
\EndIf
\EndFor
\EndFor
\end{algorithmic}
\end{algorithm}
For the subsequent discussion, it will be useful to express the value of the lower bound $L$ for an instance $P(\calT,\alpha, \beta, L, N, K)$ in a functional form. In particular, we define the function $\psi:\calD \times \calC \rightarrow \{0,1\}$ that allows us to express $L$ in another way. For nodes $v_i \in \calD, v' \in \calC$ we can define their meeting point $u \in \calT$. The function $\psi(v_i, v')$ is determined by means of Algorithm \ref{Alg:AssignPsi}, where the sequence in which we pick the nodes $v_1, \dots, v_\alpha$ is fixed. Each element of $W_{new}(u)$ can be recovered from multiple pairs of nodes that meet there. The array $\Omega(u, \delta_u)$ keeps track of the first time the file $\delta_u$ is encountered. The function $\psi(v_i, v')$ takes the value $1$ if the file $W^*$ recovered from the pair $(\dsZ(v'),\dsD(v_i))$ at $u$ belongs to $W_{new}(u)$ and has not been encountered before and 0 otherwise. A formal description is given in Algorithm \ref{Alg:AssignPsi}.
%
%
\begin{claim}
\label{clm:PsiFunction}
For an instance $P(\calT, \alpha, \beta, L, N , K)$ the following equality holds
\begin{align}
\label{eq:LrelatedtoPsifunction}
L = \sum_{i=1}^\alpha \sum_{v' \in \calC} \psi(v_i, v').
\end{align}
\end{claim}

\begin{proof}
We first note that at the end of Algorithm \ref{Alg:AssignPsi}, we have $\Omega(u, \delta_u) = 1$ for all $u \in \calT$ and all $\delta_u \subseteq W_{new}(u)$ such that  $|\delta_u| = 1$. To see this suppose that there is a $u_1 \in \calT$ and a singleton subset $\delta_{u_1}$ of $W_{new}(u_1)$ such that $\Omega(u_1, \delta_{u_1}) = 0$. Now $\delta_{u_1}$ is recovered from some delivery phase node and cache node, otherwise it would not be a subset of $W_{new}(u_1)$. As our algorithm considers all pairs of delivery phase nodes and cache nodes, at the end of the algorithm it has to be the case that $\Omega(u_1, \delta_{u_1}) = 1$.

Next, we note that for each pair $(u_1, \delta_{u_1})$ where $u_1 \in \calT$ and $\delta_{u_1}$ is singleton subset of $W_{new}(u_1)$, we can identify a unique pair of nodes $(v_i, v')$ where $v_i \in \calD$ and $v' \in \calC$ such that $\psi(v_i, v')$ and $\Omega(u_1, \delta_{u_1})$ are set to 1 at the same step of the algorithm. The remaining pairs $(v_i, v')$ that cannot be put in one to one correspondence with a pair $(u_1, \delta_{u_1})$ are such that $\psi(v_i, v')$ are set to 0.  Moreover as $\sum_{u \in \calT} \sum_{\delta_u \subseteq W_{new}(u), |\delta_u| =1} \Omega(u, \delta_u) = \sum_{u \in \calT} |W_{new}(u)| = L$, it follows that $L = \sum_{i=1}^\alpha \sum_{v' \in \calC} \psi(v_i, v')$.
\end{proof}

\begin{figure}[!t]
\centering
\includegraphics[scale=0.5]{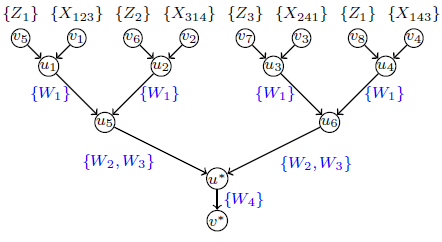}
\caption{{\small Problem instance of Example \ref{exm:PsiFunctionforNonSaturatd}. There are three users and the server contains four files.}}
\label{Fig:TreeInstanceExmpN4K3}
\end{figure}
We now illustrate the definitions introduced above by means of the following example.
\begin{example}
\label{exm:PsiFunctionforNonSaturatd}
The problem instance in Fig. \ref{Fig:TreeInstanceExmpN4K3} has seven internal nodes, $\{u_1, \ldots,u_6,u^*\}$. In the initialization step, Algorithm \ref{Alg:AssignPsi} sets $\Omega(u_i,\{W_1\}) = 0$ for $1\leq i \leq 4$, $\Omega(u_i,\{W_2\})= \Omega(u_i,\{W_3\})=0$ for $i=5,6$ and $\Omega(u^*,\{W_4\})=0$. In the next step, for node $v_1$ it sets $\psi(v_1,v_5) = 1$, $\Omega(u_1,\{W_1\})=1$ (for $v_5 \in \calC$) and $\psi(v_1,v_6) = 1$, $\Omega(u_5,\{W_2\})=1$ (for $v_6\in \calC$). For $v_7 \in \calC$ we have $\delta_{u^*} = \Delta(v_7,v_1) = \{W_3\}$ and since $W_3 \notin W_{new}(u^*)=\{W_4\}$ therefore $\psi(v_1,v_7)=0$. By the same argument we have $\psi(v_1,v_8)=0$. Thus, the contribution of $v_1$ to the lower bound, namely $\sum_{v' \in \calC} \psi(v_1, v') = 2$. The complete description of the steps after the initialization, is shown in Table \ref{tabl:ProcedureofPsiFunc}. The table should be read in column order from left to right. Within a column, the order of the operations is from top to bottom.
Note that there are two cases, $v_3 \in \calD,v_6 \in \calC$ and $v_4 \in \calD,v_6 \in \calC$ where $\psi(\cdot, \cdot)$ value is set to $0$  (since the corresponding $\Omega(\cdot, \cdot)$ values are already $1$). In both cases $\delta_{u^*}=\{W_4\}$ and since $W_4$ is recovered already, $\Omega(u^*,\{W_4\})$ has already been set to $1$ when considering $v_2 \in \calD,v_7 \in \calC$. Therefore $\psi(v_4,v_6)=\psi(v_3,v_6) = 0$. Another point to be noted is that delivery phase node $v_2$ contributes three files towards $L$ while the other delivery nodes contribute only two files each.
\end{example}

\begin{table}[h]
\centering
\begin{tabular}{|c|c|c|c|c|}
\hline
setting & $v_1$  & $v_2$  & $v_3$  & $v_4$ \\ \hline
 \multirow{3}{*}{$v_5$} & $\delta_{u_1} = W_1$ & $\delta_{u_5} = W_3$ & $\delta_{u^*} = W_2$ & $\delta_{u^*} = W_1$ \\
 & $\psi(v_1,v_5)=1$ & $\psi(v_2,v_5)=1$ & $\psi(v_3,v_5)=0$ & $\psi(v_4,v_5)=0$ \\
 & $\Omega(u_1,W_1)=1$ & $\Omega(u_5,W_3)=1$ &  &  \\ \hline
 \multirow{3}{*}{$v_6$} & $\delta_{u_5} = W_2$ & $\delta_{u_2} = W_1$ & $\delta_{u^*} = W_4$ & $\delta_{u^*} = W_4$ \\
 & $\psi(v_1,v_6)=1$ & $\psi(v_2,v_6)=1$ & $\psi(v_3,v_6)=0$ & $\psi(v_4,v_6)=0$ \\
 & $\Omega(u_5,W_2)=1$ & $\Omega(u_2,W_1)=1$ & $\Omega(u^*,W_4)=1$ & $\Omega(u^*,W_4)=1$ \\ \hline
 \multirow{3}{*}{$v_7$} & $\delta_{u^*} = W_3$ & $\delta_{u^*} = W_4$ & $\delta_{u_3} = W_1$ & $\delta_{u_6} = W_2$ \\
 & $\psi(v_1,v_7)=0$ & $\psi(v_2,v_7)=1$ & $\psi(v_3,v_7)=1$ & $\psi(v_4,v_7)=1$ \\
 &  & $\Omega(u^*,W_4)=1$ & $\Omega(u_3,W_1)=1$ & $\Omega(u_6,W_2)=1$ \\ \hline
 \multirow{3}{*}{$v_8$} & $\delta_{u^*} = W_1$ & $\delta_{u^*} = W_3$ & $\delta_{u_6} = W_2$ & $\delta_{u_4} = W_3$ \\
 & $\psi(v_1,v_8)=0$ & $\psi(v_2,v_8)=0$ & $\psi(v_3,v_8)=1$ & $\psi(v_4,v_8)=1$ \\
 &  &  & $\Omega(u_6,W_2)=1$ & $\Omega(u_4,W_3)=1$ \\ \hline
\end{tabular}
\caption{\label{tabl:ProcedureofPsiFunc}{\small The steps in Algorithm \ref{Alg:AssignPsi} after initialization when applied to Example \ref{exm:PsiFunctionforNonSaturatd}. The steps flow from the leftmost to the rightmost column, and in each column from the top to the bottom row.}}
\end{table}
\begin{corollary}
\label{corollay:sat_rem}
For an instance $P(\calT,\alpha, \beta, L, N, K)$, we have $L \leq \alpha \min (\beta, K)$. Moreover, if
$N \geq \alpha \min(\beta,K)$, then there exists an instance such that $L = \alpha \min(\beta, K)$.

\end{corollary}
\begin{proof} For a node $v_i$, where $1 \leq i \leq \alpha$, we have
\begin{align}
\label{eq:BoundOnPsi}
\sum_{v' \in \calC} \psi(v_i, v') &\leq |\cup_{v'\in \calC} \dsZ(v')| \nonumber \\
&= \hat{\beta}, \nonumber\\
& \leq \min (\beta, K).
\end{align}
Let $u$ denote the meeting point of $v'$ and $v_i$. The first inequality above holds since $\psi(v_i,v')=1$ implies that $\delta_u = \Delta(v',v_i) \subseteq W_{new}(u)$ and
$$\sum_{v' \in \calC} \psi(v_i,v') \leq |\cup_{v' \in \calC} Rec(\dsD(v_i),\dsZ(v'))| =  | Rec(\dsD(v_i),\cup_{v' \in \calC} \dsZ(v'))| \leq |\cup_{v' \in \calC} \dsZ(v')|.$$
From eq. (\ref{eq:BoundOnPsi}) we can conclude that $L = \sum_{i=1}^\alpha \sum_{v' \in \calC} \psi(v_i, v') \leq \alpha \min(\beta, K)$.
If $N \geq \alpha \min(\beta,K)$, it is easy to construct an instance with $L = \alpha \min(\beta, K)$.
We simply pick any directed tree on $\alpha + \beta$ leaves. Let the cache node indices be $Z_1$ repeated $\beta - \min(\beta, K) + 1$ times and  $Z_2, Z_3, \dots, Z_{\min(\beta, K) - 1}, Z_{\min(\beta, K)}$.
Suppose that node $v \in \calD, v' \in \calC'$ meet at node $u$. We label the delivery phase leaves such that $|\cup_{(v,v') \in \calD \times \calC'} \Delta(v',v)| = \alpha \min(\beta,K)$.
This can be done since $N$ is large enough so that we can choose the labels such that $Rec(\dsZ(v_1'), \dsD(v_1)) \cap Rec(\dsZ(v_2'), \dsD(v_2)) = \emptyset$ for $v_1',v_2' \in \calC'$ and $v_1, v_2 \in \calD$. For instance, initialize $\dsD(v)=X_{1,1,\ldots,1}$ for all $v \in \calD$ and then set $\dsD(v_i)=X_{d_1,\ldots,d_K}$, $d_j = (i-1)\alpha+j$ for $j=1,\ldots,\min(\beta,K)$, and $i =1,\ldots,\alpha$.   
\end{proof}
We illustrate the construction outlined above by means of the following example.
\begin{example}
\label{exm:trivial_saturated_instance}
Let $\alpha=\beta=2$, $K=2$, and $N=4$. We arbitrary pick a directed tree with $v_1,v_2$ as delivery nodes and $v_3,v_4$ as cache nodes. We label $\dsZ(v_3)=Z_1$ and $\dsZ(v_3)=Z_2$, and delivery nodes as $\dsD(v_1)=X_{1,2}$ and $\dsD(v_2)=X_{3,4}$. Such a problem instance is illustrated in Fig. \ref{Fig:SaturatedInstanceK2a2b2} $(a)$. As we will see later, this instance is not efficient in reusing files.
\end{example}

At this point we have established that for a given problem instance $P(\calT, \alpha,\beta,L,N,K)$, we can always generate an inequality of the form $\alpha R^{\star} + \beta M \geq L$. It is natural to therefore consider the {\it optimal} problem instances that maximize the lower bound for a given value of $\alpha,\beta,N$ and $K$.
\begin{definition}
For given $\alpha,\beta,N$ and $K$, we say that a problem instance $P(\calT^*,\alpha,\beta, L^*, N, K)$ is optimal if all problem instances $P'(\calT,\alpha,\beta, L, N, K))$ are such that $L^* \geq L$.
\end{definition}

Recall that $\hat{\beta} = |\cup_{i=1}^{\ell} \dsZ(v_i)|$.
For a problem instance $P(\calT, \alpha, \beta, L, N, K)$, it may be possible that $\hat{\beta} < \min(\beta, K)$. However, given such an instance, we can convert it into another instance where $\hat{\beta} = \min(\beta, K)$ without reducing the value of $L$. In fact the following stronger statement holds (see Appendix \ref{appendix:proofBetahat_Equal_Beta} for a proof).
\begin{claim}
\label{clm:Betahat_Equal_Beta}
For a problem instance $P(\calT, \alpha, \beta, L, N, K)$ suppose that there exists an internal node $u^*$ with associated problem instance $P^* = P(\calT_{u^*}, \alpha^*, \beta^*, L^*, N^*, K)$ such that the following condition holds. $$\hat{\beta}^* < \min(\beta^*,K).$$  Then, there exists another problem instance $P'(\calT', \alpha, \beta, L', N, K)$ where $L' \geq L$ such that the above condition does not hold.
\end{claim}
The next claim formalizes the intuitive fact that permuting the cache nodes and the delivery phase signals by the same permutation does not change the $\dsW$ labels and the lower bound of the instance.
\begin{claim}
\label{clm:bijective_mapping}
Let $P(\calT,\alpha,\beta,L,N,K)$ to be a problem instance and let $\pi:[K]\longrightarrow [K]$ to be a bijective mapping with inverse $\sigma$. Assume that the problem instance $P'(\calT',\alpha,\beta,L',N,K)$ is obtained from $P$ under the following changes for all $v \in \calD$ and $v' \in \calC$,
\begin{itemize}
\item assume $\dsZ^{P}(v)=Z_i$, then set $\dsZ^{P'}(v)=Z_{\pi(i)}$,
\item assume $\dsD^{P}(v)=X_{d_1,\ldots,d_K}$, then set $\dsD^{P'}(v)=X_{d_{\sigma(1)},\ldots,d_{\sigma(K)}}$,
\end{itemize}
then $W_{new}^{P'}(u)=W_{new}^{P}(u)$, $\dsW^{P'}(u)=\dsW^{P}(u)$ for $u\in\calT$, and $L'=L$.
\end{claim}
\begin{proof}
We note that
$$Rec(Z_i,X_{d_1,\ldots,d_K})=W_{d_i} = W_{d_{\sigma(\pi(i))}} = Rec(Z_{\pi(i)},X_{d_{\sigma(1)},\ldots,d_{\sigma(K)}})$$
for $i=1,\ldots,K$. Therefore, for any $v \in \calD$ and $v' \in \calC$, we have $\Delta^{P'}(v',v) =\Delta^{P}(v',v)$ and more generally $\Delta^{P'}(u,u) = \Delta^{P}(u,u)$. From this and that $\dsW^P(u)=\Delta^{P}(u_l,u_l) \cup \Delta^{P}(u_r,u_r)$, we have $\dsW^P(u)=\dsW^{P'}(u)$ for any $u \in \calT$. Using eq. (\ref{eq:wnew_def}), we have $W^{P'}_{new}(u)=W^{P}_{new}(u)$ for all $u \in \calT'$. Consequently, it follows that $L'=L$.
\end{proof}

Henceforth, we will assume w.l.o.g. that $\hat{\beta} = \min(\beta, K)$ and that Claim \ref{clm:Betahat_Equal_Beta} holds.
Our next lemma shows a structural property of problem instances. Namely for an instance where $L < \alpha \min(\beta,K)$, increasing the number of files allows us to increase the value of $L$. This lemma is a key ingredient in our proof of the main theorem (the proof appears in the Appendix).
\begin{lemma}
\label{lemma:Increase_N_Increase_L}
Let $P=P(\calT,\alpha,\beta,L,K,N)$ be an instance where $L < \alpha \min(\beta,K)$.
Then, we can construct a new instance $P'=P(\calT',\alpha,\beta,L',K,N+1)$, where $L'=L+1$.
\end{lemma}

Informally, another property of optimal problem instances is that the same file is recovered as many times as possible at the same level of the tree. For instance, in Fig. \ref{Fig:TreeInstanceExmp} $W_1$ is recovered in both $\calT_{u^*(l)}$ and $\calT_{u^*(r)}$. In fact, intuitively it is clear that the same set of files can be reused in any subtrees of an internal node. Our next claim formalizes this intuition. Recall that for a node $u$, $\Gamma_l = \cup_{v \in \calT_{u(l)}} W_{new}(u)$ and $\Gamma_r = \cup_{v \in \calT_{u(r)}} W_{new}(u)$.
\begin{claim}
\label{clm:Gamma_right_subset_Gamma_left}
Consider an instance $P=P(\calT,\alpha,\beta,L,K,N)$. For all nodes $u \in \calT$, suppose w.l.o.g. that $|\Gamma_l| \geq |\Gamma_r|$. Suppose that there exist a node $u \in \calT$ such that such that $\Gamma_r \nsubseteq \Gamma_l$. Then there exists another instance $P'(\calT', \alpha, \beta, L', N', K)$ such that $N' \leq N$, $L' \geq L$, and $\Gamma_r \subseteq \Gamma_l$ for all $u \in \calT'$.
\end{claim}
%
Next, we upper bound the maximum value of $|W_{new}(u)|$ for a node $u \in \calT$.
\begin{figure}[!t]
\centering
\includegraphics[scale=0.5]{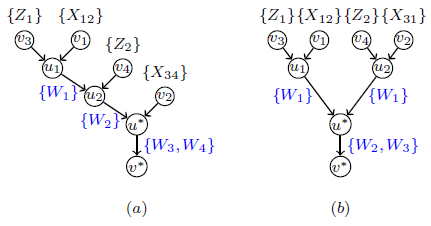}
\caption{{\small $(a)$ Problem instance $P'(\calT', \alpha,\beta,L,N',K)$, $(b)$ problem instance $P(\calT, \alpha,\beta,L,N,K)$ where $\alpha=2$, $\beta=2$ and $K=2$. Both instances reach $L=\alpha \min(\beta,K)=4$ with different number of files $N=3$ and $N'=4$.}}
\label{Fig:SaturatedInstanceK2a2b2}
\end{figure}
\begin{claim}
\label{clm:WnewBound}
In instance $P(\calT, \alpha, \beta, L, N , K)$, consider an internal node $u$. Let $\rho(u) = \hat{\alpha}_l [\min ( \beta_r, K - \beta_l)]^{+} + \hat{\alpha}_r [\min( \beta_l, K-\beta_r)]^{+}$. We have
\begin{align*}
|W_{new}(u)| \leq
 \min \left( \rho(u),
N - |\Gamma_l \cup \Gamma_r| \right).
\end{align*}
\end{claim}
\begin{proof}
From eq. (\ref{eq:wnew_2}) it follows that
\begin{align*}
|W_{new}(u)| &\leq  |\Delta_{rl}(u)\setminus \dsW(u)| +  |\Delta_{lr}(u)\setminus \dsW(u)|.
\end{align*}
Next, we observe that
\begin{align*}
&|\Delta_{rl}(u)\setminus \dsW(u)| = |Rec(\dsZ(u_r) \setminus \dsZ(u_l), \dsD(u_l)) \setminus \dsW(u)| \\
&\leq |\dsD(u_l)| \times |\dsZ(u_r) \setminus \dsZ(u_l)| \\
&\stackrel{(a)}{\leq} \hat{\alpha}_l \times \min ( \hat{\beta}_r, K - \hat{\beta}_l),\\
&\stackrel{(b)}{=} \hat{\alpha}_l \times [\min ( \beta_r, K - \beta_l)]^+,
\end{align*}
where inequality (a) holds, since $|\dsD(u_l)| = \hat{\alpha}_l$ and $|\dsZ(u_r) \setminus \dsZ(u_l)| \leq \min ( \hat{\beta}_r, K - \hat{\beta}_l)$. Inequality (b) holds under the conditions $\hat{\beta}_l = \min(\beta_l,K)$ and $\hat{\beta}_r = \min(\beta_r,K)$ (see Claim \ref{clm:app_rho_eq_rho_tild} in Appendix). We can bound $|\Delta_{lr}(u)\setminus \dsW(u)|$ in a similar manner.

To conclude the proof we note that instances $P_l$ and $P_r$ recover a total of $|\Gamma_l \cup \Gamma_r|$ sources. As the total number of sources is $N$, $|W_{new}(u)| \leq N - |\Gamma_l \cup \Gamma_r|$.
\end{proof}

\begin{definition} {\it Saturation number.}
Consider an instance $P^*(\calT^*, \alpha, \beta, L^*, N^* , K)$, where $L^* = \alpha \min(\beta, K)$, such that for all problem instances of the form $P(\calT, \alpha, \beta, L^*, N,K)$, we have $N^* \leq N$. We call $N^*$ the saturation number of instances with parameters $(\alpha, \beta, K)$ and denote it by $N_{sat}(\alpha, \beta, K)$.
\end{definition}
In essence, for given $\alpha, \beta$ and $K$, saturated instances are most efficient in using the number of available files. It is easy to see that $N_{sat}(\alpha, \beta, K) \leq \alpha \min(\beta, K)$ since one can construct an instance with lower bound $\alpha \min(\beta,K)$ when $\alpha \min (\beta,K) \leq N$  (see  Corollary \ref{corollay:sat_rem}).

\begin{example}
Consider the two problem instances $P$ and $P'$ with $\alpha=2,\beta=2$ and $K=2$ that are shown in Fig. \ref{Fig:SaturatedInstanceK2a2b2}. The lower bound for both instances is $L=\alpha \min(\beta,K) = 4.$ However, instance $P$ uses one file less than $P'$. This reduction is accomplished by reusing file $W_1$ at both $\calT_{u^*(l)}$ and $\calT_{u^*(r)}$. The instance $P'$ can be treated as trivial instance constructed by the procedure suggested in the proof of Corollary \ref{corollay:sat_rem} as it uses $N'=\alpha \min(\beta,K) = 4$ files. It can be verified by exhaustive search that $P$ is one of the problem instances associated with $N_{sat}(2,2,2)$; therefore, $N_{sat}(2,2,2) = 3$.
\end{example}

\begin{definition} {\it Atomic problem instance.}
For a given optimal problem instance $P(\calT, \alpha, \beta, L, N, K)$ it is possible that there exist other optimal problem instances $P_i(\alpha_i, \beta_i, L_i, N, K), i = 1, \dots, m$ with $m \geq 2$ such that $\sum_{i=1}^m \alpha_i = \alpha, \sum_{i=1}^m \beta_i = \beta$ and $\sum_{i=1}^m L_i = L$, i.e., the value of $L$ follows from appropriately combining smaller problems. In this case we call the instance $P$ as non-atomic. Conversely, if such smaller problem instances do not exist, we call $P$ an atomic problem instance.
\end{definition}

\begin{figure}[!t]
\centering
\includegraphics[scale=0.5]{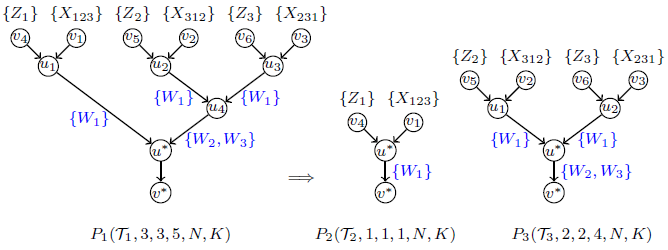}
\caption{{\small Problem instances with $N=K=3$. Instance $P_1$ is non-atomic as the corresponding lower bound can be obtained by summing the lower bounds from $P_2$ and $P_3$.}}
\label{Fig:AtomicInstance}
\end{figure}

\begin{example}
Consider the problem instance $P_1$ shown in Fig. \ref{Fig:AtomicInstance} with $N=K=3$. The lower bound associated with this instance, $3R^\star+3M \geq 5$, can be obtained by combining the lower bounds acquired by $P_2$ and $P_3$. Specifically, instance $P_2$ yields $R^\star+M \geq 1$ and instance $P_3$ yields $2R^\star+2M \geq 4$. Note that in $P_1$ the last edge $(u^*,v^*)$ is such that $W_{new}(u^*) =\emptyset$. Thus, the tree can be split in two separate instances at $u^*$. Thus it is non-atomic.
\end{example}

It is evident that instances where no new file is recovered in the last edge are non-atomic. However, we emphasize that there are other instances that are non-atomic as well. For example, consider instance $P_1'$, obtained from $P_1$ where we change the label $\dsD(v_3)$ to $X_{221}$. In $P_1'$, the labels of edges $(u_4,u^*)$ and $(u^*,v^*)$ will change to $\{W_2\}$ and $\{ W_3 \}$ respectively; none of the other labels will change. Even though $W_{new}(u^*)$ is nonempty in $P_1'$, but we still call it non-atomic since the associated lower bound does not change.


The following theorem and its corollary are the main results of our paper and can be used to identify optimal problem instances. 
\begin{theorem}
\label{thm:Lvalue}
Suppose that there exists an optimal and atomic problem instance $P_o(\calT = (V, A), \alpha, \beta, L_o, N, K)$. Then, there exists an optimal and atomic problem instance $P^*(\calT^* = (V^*, A^*), \alpha, \beta, L^*, N, K)$ where $L^*=L_o$ with the following properties. Let us denote the last edge in $P^*$ with $(u^*, v^*)$. Let $P^*_l = P(\calT^*_{u^*(l)},\alpha_l, \beta_l,L^*_l,|\Gamma_l|,K)$ and $P^*_r = P(\calT^*_{u^*(r)},\alpha_r, \beta_r,L^*_r,|\Gamma_r|,K)$. Then, we have 
\begin{eqnarray}
L^*_l &=& \alpha_l \min (\beta_l, K),\nonumber \\
L^*_r &=& \alpha_r \min (\beta_r, K), \text{~and} \nonumber \\
L^* &=&  \min\left(\alpha \min(\beta, K), L^*_l + L^*_r + N - N_0 \right), \label{eq:cap_lower_bound}
\end{eqnarray}
where $N_0 = \max(N_{sat}(\alpha_l, \beta_l, K), N_{sat}(\alpha_r, \beta_r, K))$\footnote{As the instance is atomic, we have $N > N_0$.}. Furthermore, $\min (\beta_l, \beta_r) < K$.
\end{theorem}

\begin{proof}
Note that we assume that the problem instance $P_o$ is atomic. This implies that $W_{new}^{P_o}(u^*) \neq \emptyset$ and, consequently, $N > |\Gamma_l|,|\Gamma_r|$. Using Claim \ref{clm:Betahat_Equal_Beta} we can assert that $\hat{\beta}_l = \min(\beta_l,K)$ and $\hat{\beta}_r = \min(\beta_r,K)$.


%
%
We denote by $(u^*,v^*)$, the last edge in $P_o$. We let $P_l=P(\calT_{u^*(l)}, \alpha_l,\beta_l,L_l,|\Gamma_l|, K)$ and $P_r=P(\calT_{u^*(r)}, \alpha_r,\beta_r,L_r,|\Gamma_r|, K)$. It is easy to see that $L_o = L_l + L_r+|W_{new}^{P_o}(u^*)|$.
Suppose that $L_l < \alpha_l \min (\beta_l, K)$. We apply the result of Lemma \ref{lemma:Increase_N_Increase_L}, by noting that $|\Gamma_l| < N$, and conclude that there exists another instance $P_l^{**} = P(\calT^{**}_{u^*(l)},\alpha_l, \beta_l,L^*_l+1,|\Gamma_l| +1,K)$ that can replace $P_l$, where the new file is denoted $W^*$. We also note that in $P_o$, $W^* \in W_{new}^{P_o}(u^*)$. Let us denote the new instance $P'_o$.
We emphasize that the nature of the modification in Lemma \ref{lemma:Increase_N_Increase_L} is such that $\Delta^{P'_o}(u^*, u^*) = \Delta^{P_o}(u^*, u^*)$. Moreover, we note that $\dsW^{P'_o}(u^*) = \dsW^{P_o}(u^*) \cup \{W^*\}$. Thus,
\begin{align*}
&W_{new}^{P'_o}(u^*) = \Delta^{P'_o}(u^*, u^*) \setminus \dsW^{P'_o}(u^*)\\
&= \Delta^{P'_o}(u^*, u^*) \setminus \dsW^{P_o}(u^*) \cup \{W^*\}\\
&= W^{P_o}_{new} (u^*) \setminus \{W^*\}.
\end{align*}
The problem instance $P_o'$ is also optimal since $L_l$ is increased by one and $|W_{new}^{P_o}(u^*)|$ is decreased by one, leaving $L_o$ unchanged. Therefore, moving files from $W_{new}^{P_o}(u^*)$ to either $P_l$ or $P_r$ preserves optimality. In addition, from $L_o'=L_o$ and that $P_o$ is atomic, $P_o'$ is atomic.
Based on this argument, we can immediately conclude that we cannot have $L_l < \alpha_l \min (\beta_l, K)$ and $L_r < \alpha_r \min (\beta_r, K)$ as the file $W^*$ can be used to simultaneously modify the instance $P_r$. Upon this modification, we can conclude that $L_o$ can be increased by one, which contradicts the optimality of the instance $P_o$. Thus we assume that $L_r = \alpha_r \min (\beta_r, K)$. We can repeatedly apply the operation of moving files from $W^{P_o}_{new}(u^*)$ to $P_l$ until we have $L^*_l = \alpha_l \min (\beta_l, K)$. It has to be the case that $|W^{P_o}_{new}(u^*)| > \alpha_l \min (\beta_l, K) - |\Gamma_l|$ so that we can repeatedly apply the operation of moving the files, for if this were not true, the instance $P_o$ would not be atomic.

We will denote the instance that we arrive at after completing these modification by $P^*$ which is optimal and atomic.
We can also observe at this point that if we have $\beta_l \geq K$ and $\beta_r \geq K$ so that $\hat{\beta}_l = \hat{\beta}_r = K$, then $W_{new}^{P^*}(u^*) = \emptyset$ (by Claim \ref{clm:WnewBound}) which implies that the original instance $P_o$ is not atomic. Thus, either $\beta_l$ or $\beta_r$ or both have to be strictly smaller than $K$. In the discussion below we assume w.l.o.g. that $\beta_r < K$.
\begin{figure}[!t]
\centerline{
\includegraphics[scale=0.5]{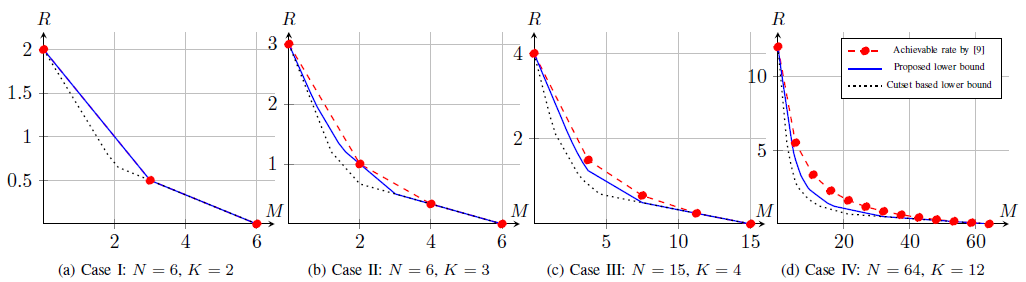}
}
\caption{{\small Comparison of the proposed lower bound and the cutset bound.}}
\label{Fig:suboptimalboound}
\end{figure}
It is easy to see that
\begin{align*}
L^* = L^*_l + L^*_r + |W^{P^*}_{new}(u^*)|.
\end{align*}

We define $\tilde{\rho}(u^*) = \alpha_l \times [\min ( \beta_r, K - \beta_l)]^{+} + \alpha_r \times [\min(\beta_l, K-\beta_r)]^{+}$ where $\tilde{\rho}(u^*) \geq \rho(u^*)$ due to the fact that $\alpha_l \geq \hat{\alpha}_l $ and $\alpha_r \geq \hat{\alpha}_r $. Using this and Claim \ref{clm:WnewBound}, we have that
\begin{align*}
|W^{P^*}_{new}(u^*) | &\leq \min\left(\tilde{\rho}(u^*) , N - \max (|\Gamma^*_l|, |\Gamma^*_r|)\right).
\end{align*}
For an optimal instance, we claim that the above inequality is met with equality. If $L^* = \alpha\min(\beta,K)$ there is nothing to prove. In this case, $|W^{P^*}_{new}(u^*) |=\alpha\min(\beta,K) - L^*_l-L^*_r=\tilde{\rho}(u^*)$ (see Claim \ref{clm:ab_eq_albl_plus_arbr_plus} in Appendix) and the above inequality is met with equality.

Otherwise, we have $L^* < \alpha \min(\beta,K)$ which implies $\tilde{\rho}(u^*) > |W^{P^*}_{new}(u^*)|$ and $\tilde{\rho}(u^*) > N- \max (|\Gamma^*_l|, |\Gamma^*_r|)$. From the Claim \ref{clm:Gamma_right_subset_Gamma_left}, we can assume that either $\Gamma^*_l \subseteq \Gamma^*_r$ or $\Gamma^*_r \subseteq \Gamma^*_l$. In $P^*$, $N_{used}= \max \left(|\Gamma^*_l|, |\Gamma^*_r|\right)+|W^{P^*}_{new}(u^*)|$ files are used so far. Now, if $N > N_{used}$, we can use Lemma \ref{lemma:Increase_N_Increase_L} to conclude that there exists a problem instance $P''(\calT'', \alpha,\beta, L'',N'',K)$ where $N'' = N_{used} +1 \leq N$ and $L'' = L^*+1$. This is a contradiction since we assumed that $P^*$ is optimal. Therefore, $N \leq N_{used}$. In addition, since the number of available files is $N$ thus $N \geq N_{used}$. As a result, $N = N_{used} = \max \left(|\Gamma^*_l|, |\Gamma^*_r|\right)+|W^{P^*}_{new}(u^*)|$ and the inequality is met with equality. In both cases, we conclude that
\begin{align*}
|W^{P^*}_{new}(u^*) | = \min\left(\tilde{\rho}(u^*) , N - \max (|\Gamma^*_l|, |\Gamma^*_r|)\right).
\end{align*}
%
%
It follows that
\begin{align*}
L^* = \min\left(\alpha \min(\beta, K), L_l^*+L_r^*+N-\max(|\Gamma^*_l|, |\Gamma^*_r|) \right).
\end{align*}
If $L^* = \alpha \min(\beta,K)$ the saturated instance associated to $N_{sat}(\alpha,\beta,K)$ is an optimal instance. Otherwise, $L^* < \alpha \min(\beta, K)$, we have
\begin{align}
\label{eq:BoundonWnewTheorem}
|W^{P^*}_{new}(u^*)| &= N - \max (|\Gamma^*_l|, |\Gamma^*_r|) \\
&\leq N -  \max (N_{sat}(\alpha_l,\beta_l, K), N_{sat}(\alpha_r, \beta_r, K)). \nonumber
\end{align}
We claim that for $P^*$ to be optimal, $P^*_l$ and $P^*_r$ have to be such that $\max(|\Gamma^*_l|,|\Gamma^*_r|) = \max (N_{sat}(\alpha_l,\beta_l, K), N_{sat}(\alpha_r, \beta_r, K))$. To see this we proceed as follows. Note that by the definition of saturation number, there exist problem instances $P'_l(\calT'_l, \alpha_l, \beta_l, L'_l, N'_l, K)$ and $P'_r(\calT'_r, \alpha_r, \beta_r, L'_r, N'_r, K)$ such that $L'_l = L^*_l$, $L'_r = L^*_r$, $N'_l = N_{sat}(\alpha_l,\beta_l,K)$ and $N'_r = N_{sat}(\alpha_r,\beta_r,K)$. W.l.o.g. let assume $N'_l \geq N'_r$. By the Claims \ref{clm:Betahat_Equal_Beta} and \ref{clm:Gamma_right_subset_Gamma_left} problem instances $P'_l$ and $P'_r$ can be modified in such a way that $\hat{\beta}'_l = \min(\beta_l,K)$, $\hat{\beta}'_r = \min(\beta_r,K)$ and $\Gamma'_l \subseteq \Gamma'_r$. Also, by Claim \ref{clm:bijective_mapping} we can set $\cup_{v \in \calC'_l}\dsZ(v) = \{ Z_1, \ldots, Z_{\hat{\beta}'_l} \}$ and $\cup_{v \in \calC'_r}\dsZ(v) = \{Z_{K-\hat{\beta}'_r+1}, \ldots,Z_K \}$. This ensures that $\hat{\beta}_l = \min(\beta_l,K)$, $\hat{\beta}_r = \min(\beta_r,K)$, and $\hat{\beta} = \min(\beta,K)$ hold in the defined problem instance.
Now, consider the problem instance $P'=P(\calT',\alpha, \beta, L', N, K)$ with last edge $(u',v')$ where $P'_l$ and $P'_r$ are instances of $u'_l$ and $u'_r$ respectively. The instance $P'$ uses $N'_l+|W^{P'}_{new}(u')|$ files. If $N-N'_l-|W^{P'}_{new}(u')| \geq 1$, then we are able to apply Lemma \ref{lemma:Increase_N_Increase_L} $N-N'_l-|W^{P'}_{new}(u')|$ times and come up with a modified version of $P'$ so that either $L'=\alpha \min(\beta,K)$ or $N-N'_l-|W^{P'}_{new}(u')|=0$. The first case cannot happen since by assumption $P^*$ is optimal and $L' \leq L^* < \alpha\min(\beta,K)$. Therefore, $|W^{P'}_{new}(u')|= N-N'_l$ and $L'=L^*_l+L^*_r+N-N'_l$. Finally, as $L' \leq L^*$ and $L^* \leq L^*_l+L^*_r+N-N'_l$, we conclude that $L' = L^*$.
%
\end{proof}

\begin{corollary}
\label{corollary:Lvalue}
Suppose that there exists an optimal and atomic problem instance $P_o(\calT = (V, A), \alpha, \beta, L_o, N, K)$. Consider problem instances $P'_l(\alpha'_l,\beta'_l,L'_l,N,K)$ and $P'_r(\alpha'_r,\beta'_r,L'_r,N,K)$ such that $\alpha'_l + \alpha'_r = \alpha$ and $\beta'_l+\beta'_r = \beta$ such that $N \geq N'_0 = \max(N_{sat}(\alpha'_l,\beta'_l,K),N_{sat}(\alpha'_r,\beta'_r,K))$. Then we have
$$L_o \geq \min\left(\alpha \min(\beta, K), L'_l + L'_r + N - N'_0) \right).$$
\end{corollary}
\begin{proof}
The result follows by applying the arguments in the proof of Theorem \ref{thm:Lvalue}, to the problem instance where $P^*_l$ and $P^*_r$ are replaced by $P'_l$ and $P'_r$ respectively.
\end{proof}
The following example demonstrates the effectiveness of Corollary \ref{corollary:Lvalue}.
\begin{example}
Consider a system with $N=64$, $K=12$ and cache size $M = 16/3$. The cut-set bound for such a system provides a lower bound $R^{\star}(M) \geq 77/27=2.852$. Now, using the approach of Theorem \ref{thm:Lvalue} for $\alpha = 12$, $\beta=8$, $(\alpha_l,\beta_l) = (\alpha_r,\beta_r) =(6,4)$ yields $12R^{\star}+8M \geq \min(12 \times 8, 24+24+64-N_{sat}(6,4,12))$. 
It can be shown that $N_{sat}(6,4,12) \leq 17$ (see Algorithm \ref{Alg:SubOptimalInstance} below). Therefore, $R^{\star}(M) \geq 157/36=4.361$. This is significantly closer to the achievable rate of $5.5$ (from \cite{maddahN14}).
\end{example}

Theorem \ref{thm:Lvalue} can be leveraged effectively if it can also yield the optimal values of $\alpha_l, \beta_l$ and $\alpha_r, \beta_r$. However, currently we do not have an algorithm for picking them in an optimal manner. Moreover, we also do not have an algorithm for finding $N_{sat}(\alpha, \beta, K)$. Thus, we have to use Corollary \ref{corollary:Lvalue} with an appropriate upper bound on $N_{sat}(\alpha, \beta,K)$ in general. 

Algorithm \ref{Alg:SubOptimalInstance} in Section \ref{sec:subsec_anlytic_bound_Nsat} provides a constructive algorithm for upper bounding $N_{sat}(\alpha, \beta, K)$. Setting $\alpha_l = \lceil \alpha /2 \rceil$, $\beta_l = \lfloor \beta /2 \rfloor$ in Theorem \ref{thm:Lvalue} and applying this approach to upper bound the saturation number, we can obtain the results plotted in Fig. \ref{Fig:suboptimalboound}.
\subsection{An analytic bound on the saturation number}
\label{sec:subsec_anlytic_bound_Nsat}
Recall that the saturation number for a given $\alpha, \beta$ and $K$ is the minimum value of $N$ such that there exists a problem instance $P(\calT, \alpha, \beta, L, N, K)$ with $L = \alpha \min (\beta, K)$.
In particular, this implies that if we are able to construct a problem instance with $N'$ files with a lower bound equal to $\alpha \min (\beta, K)$, then, $N_{sat}(\alpha,\beta, K) \leq N'$. In Algorithm \ref{Alg:SubOptimalInstance}, we create one such problem instance.

\begin{algorithm}[!t]
\caption{Instance construction for upper bounding $N_{sat}(\alpha,\beta,K)$}\label{Alg:SubOptimalInstance}
\algrenewcommand\algorithmicrequire{\textbf{Input:}}
\algrenewcommand\algorithmicfunction{\textbf{Initialization}}
\begin{algorithmic}[1]
\small{
\Require $\alpha$, $\beta$ and $K$.
\Function{}{}
\State Let $(u^*, v^*)$ be last edge and set $U_{new}= \{ u^* \}$.
\State \label{alglin:Set_aus_bus}Set $\dsZ(u^*) = \{Z_1, Z_2, \dots, Z_{\min(\beta,K)}\}$ and $b(u^*)=\beta$,  $a(u^*)=\alpha$.
\State $\calC = \emptyset$ and $\calD = \emptyset$.
\EndFunction
\Procedure{Tree Construction \& Cache nodes labeling}{}
\While{$U_{new}$ is nonempty}
\State Pick $u \in U_{new}$, create nodes $u_l$ and $u_r$,  edges $(u_l,u)$ and $(u_r,u)$, add them to $\calT_0$.
\State \label{alglin:Set_au_bu}Set $a(u_l) = \lceil a(u)/2 \rceil$, $b(u_l) = \lfloor b(u)/2 \rfloor$ and  $a(u_r) = a(u) - a(u_l)$, $b(u_r) = b(u) - b(u_l)$.
\State Set $\dsZ(u_l)$ and $\dsZ(u_r)$ be subsets of $\dsZ(u)$ of sizes $\min(b(u_l), K)$ and $\min(b(u_r), K)$ respectively with minimum intersection.
\State Remove $u$ from $U_{new}$.
\If{$a(u_l)+ b(u_l) \geq 2$} 
\State Add $u_l$ to $U_{new}$.
\Else
\State If $b(u_l)==1$ add $u_l$ to $\calD$ otherwise to $\calC$.
\EndIf
\If{$a(u_r)+ b(u_r) \geq 2$}
\State Add $u_r$ to $U_{new}$.
\Else
\State If $b(u_r)==1$ add $u_r$ to $\calD$ otherwise to $\calC$.
\EndIf
\EndWhile
\EndProcedure
\Procedure{Delivery nodes labeling}{}
\State Let $\calD = \{v_{1}, \ldots, v_{\alpha} \}$.
\For{$r = 1, \ldots, \min(\beta,K)$}
\State Pick a node $v \in \calC$ with $\dsZ(v)=\{Z_r\}$ and denote it by $v_{r+\alpha}$.
\EndFor
\State Let $\calC \setminus \{v_{\alpha+1}, \ldots,  v_{\alpha+\min(\beta,K)}\} = \{v_{\alpha+\min(\beta,K)+1}, \ldots,v_{\beta} \}$.
\For{$t = 1, \ldots, \alpha$}
\For{$r = 1, \ldots, \min(\beta,K)$}
\State \label{alglin:dlv_d_ass}$d_{r} = (t-1)\min(\beta,K)+r$.
\EndFor
\For{$r = \min(\beta,K) + 1, \ldots, K$}
\State $d_{r} = 1$.
\EndFor
\State Set $\dsD(v_{t}) = X_{d_1,\ldots,d_K}$
\EndFor
\EndProcedure
\Procedure{Modify Delivery phase signals}{}
\State \label{alglin:zero_instance} Denote current instance by $P_0(\calT_0,\alpha, \beta, L_0, N_0,K)$.
\State \label{alglin:dlv_clm_gamma}Modify $P_{0}(\calT_{0},\alpha, \beta,L_{0},N_{0},K)$ by Claim \ref{clm:Gamma_right_subset_Gamma_left} to obtain $P(\calT,\alpha, \beta, L, \hat{N}_{sat},K)$.
\EndProcedure
\algrenewcommand\algorithmicrequire{\textbf{Output:}}
\Require $\hat{N}_{sat}(\alpha, \beta, K) = |\Gamma(v^*)|$, $P(\calT,\alpha,\beta,L,\hat{N}_{sat},K) $.
}
\end{algorithmic}
\end{algorithm}

The basic idea of Algorithm \ref{Alg:SubOptimalInstance} is as follows. 
The first part focuses on the construction of the tree, without labeling the leaves.
For a given $\alpha$ and $\beta$, we first initialize a tree that just consists of a single edge $(u^*, v^*)$. Following this, we partition $\alpha$ into two parts $\alpha_l = \lceil \alpha/2 \rceil$ and $\alpha_r = \alpha - \alpha_l$. On the other hand, $\beta$ is split into $\beta_l = \lfloor \beta/2 \rfloor $ and $\beta_r = \beta - \beta_l$. The algorithm, then recursively constructs the left and right subtrees of $u^*$.
It is important to note that the split in the $(\alpha, \beta)$ pair is done in such a manner that each subtree gets the floor and the ceiling of the one of the quantities. Moreover, the labeling of the cache node leaves is such that for a given node $u$, $|\dsZ(u_l) \cap \dsZ(u_r)|$ is as small as possible. The underlying reason for such a labeling is to ensure that the condition of Claim \ref{clm:Betahat_Equal_Beta} doesn't hold for any $u \in \calT$.

Following, the construction of the tree, the second phase of the algorithm labels each of the delivery phase nodes, so that the computed lower bound is $L = \alpha \beta$. In this step we use $N=\alpha\beta$ files (see the procedure discussed in the proof of Corollary \ref{corollay:sat_rem}). In the third and final phase of the algorithm we modify the instance so that for any node $u \in \calT$, we have that either $\Gamma_l \subseteq \Gamma_r$ or $\Gamma_r \subseteq \Gamma_l$; we use  Claim \ref{clm:Gamma_right_subset_Gamma_left} to achieve this. 
In the beginning all recovered files in the constructed instance are distinct so that $\Gamma(u_l) \cap \Gamma(u_r) = \emptyset$ for all nodes $u$. W.l.o.g. assume that $|\Gamma(u_r)| \leq |\Gamma(u_l)|$. An application of Claim \ref{clm:Gamma_right_subset_Gamma_left} will thus cause a significant reduction in the number of files that are used. The following lemma quantifies this reduction.
\begin{lemma}
\label{lem:UpBoundNsat} For given $\alpha$, $\beta$ and $K$ if $\beta \leq K$ then,
\beqa
N_{sat}(\alpha, \beta, K) \leq \left\lfloor \frac{2\alpha \beta + \alpha + \beta}{3} \right\rfloor.
\eeqa
\end{lemma}
\begin{proof}
We use Algorithm \ref{Alg:SubOptimalInstance} to generate problem instance $P(\calT, \alpha, \beta, L, \hat{N}_{sat},K)$ so that $L = \alpha \beta$. By the definition of the saturation number we have $N_{sat}(\alpha, \beta, K) \leq \hat{N}_{sat} $ hence we just need to show that $\hat{N}_{sat} \leq \frac{2\alpha \beta + \alpha + \beta}{3}$.

First, we need to show that $L=\alpha \beta$. By line \ref{alglin:dlv_d_ass} of the algorithm the file $W_{(t-1)\beta+r}$ is recoverable in instance $P_0$ by the pair $(\dsD(v_t),\dsZ(v_{\alpha+r}))$ or equivalently $\Delta(v_t,v_{\alpha+r})=W_{(t-1)\beta+r}$ for $1\leq t \leq \alpha$ and $1 \leq r \leq \beta$. On the other hand, $\dsW(v^*) = \cup_{t=1}^\alpha \cup_{r=1}^\beta \Delta(v_t,v_{\alpha+r}) $ therefore $\dsW(v^*) = \{W_1, \ldots, W_{\alpha\beta} \}$. Recall that $\dsW(v^*) = \cup_{u \in \calT_0} W_{new}(u)$ and $L_0=\sum_{u \in \calT_0}|W_{new}(u)|$ so we have $L_0 \geq |\dsW(v^*)| = \alpha \beta$. But $L_0 \leq \alpha \beta$, by Corollary \ref{corollary:Lvalue}, therefore $L_0 = \alpha \beta$. In phase III of the Algorithm (Modify Delivery Phase Signals) using Claim \ref{clm:Gamma_right_subset_Gamma_left}, we have $L \geq L_{0}$ and since $L \leq \alpha \beta$ and $L_0 = \alpha \beta$ thus $L = \alpha \beta$.

W.l.o.g we set left incoming node such that $\Gamma(u_r) \subseteq \Gamma(u_l)$. Starting from the root node $v^*$, we let the set $\{u_{0}, u_{1}, \ldots,u_{t} \}$ and $\{w_0, \ldots, w_{t-1} \}$ to be the left and right incoming nodes respectively so that $u_{i}$ is topologically higher than $u_{j}$ for $i < j$, $u_{t} = u^*$ and $u_{0}$ to be a leaf. This is depicted in Fig. \ref{Fig:SaturationPath}. Recall that $\Gamma(u) = W_{new}(u) \cup \Gamma(u_l) \cup \Gamma(u_r)$ and $W_{new}(u) \cap \left( \Gamma(u_l) \cup \Gamma(u_r) \right) = \emptyset$ for any $u \in \calT$. Therefore, recursively we have,
\begin{eqnarray}
\label{eq:Nhat_satGamma}
\hat{N}_{sat} &=& |\Gamma(v^*)| = |\Gamma(u_t)|, \nonumber \\
&=& |W_{new}(u_t)| + |\Gamma(u_{t-1})|, \nonumber \\
&=& \sum_{i=1}^t |W_{new}(u_i)|,
\end{eqnarray}
where we used $W_{new}(u_0) = \emptyset$ since $u_0$ is a leaf.

In Algorithm \ref{Alg:SubOptimalInstance}, $a(u)$ and $b(u)$ denote the number of delivery phase nodes and the number cache nodes, respectively in the subtree rooted at $u$.
Note that by definition, we have
\begin{align*}
L &= |W_{new}(u_t)|+\sum_{u \in \calT_{u_{t-1}}}|W_{new}(u)|+\sum_{u \in \calT_{w_{t-1}}}|W_{new}(u)|.
\end{align*}
Using Corollary \ref{corollary:Lvalue} we conclude that $\sum_{u \in \calT_{u_{t-1}}}|W_{new}(u)| \leq a(u_{t-1})b(u_{t-1})$ and $\sum_{u \in \calT_{w_{t-1}}}|W_{new}(u)| \leq a(w_{t-1}) b(w_{t-1})$.
Similarly, using Claim \ref{clm:WnewBound}, we have that $|W_{new}(u_t)| \leq a(u_{t-1})b(w_{t-1}) + a(w_{t-1})b(u_{t-1})$. In fact, all these inequalities are met with equality. This can be seen as follows. An application of Claim \ref{clm:Gamma_right_subset_Gamma_left} does not change the lower bound, which implies that $L = \alpha\beta = a(u_t)b(u_t)$. But, $a(u_t) = a(u_{t-1}) + a(w_{t-1})$ and $b(u_t) =  b(u_{t-1}) + b(w_{t-1})$ so that
\begin{align*}
L &= a(u_{t-1})b(w_{t-1}) + a(w_{t-1})b(u_{t-1}) + a(u_{t-1})b(u_{t-1}) + a(w_{t-1})b(w_{t-1}).
\end{align*}
An inductive argument can be made to show a similar result for $u_i$, $i = 1, \dots, t-1$.

Using these results and the equality in (\ref{eq:Nhat_satGamma}) yields,
\begin{eqnarray}
\label{eq:Nhat_satAlphBet}
\alpha \beta &=& L, \nonumber \\
&=& \sum_{u \in \calT} |W_{new}(u)|,\nonumber \\
&=& \sum_{i=0}^t |W_{new}(u_i)| + \sum_{i=0}^{t-1} \sum_{u \in \calT_{w_i}} |W_{new}(u)|, \nonumber \\
&=& \hat{N}_{sat} + \sum_{i=0}^{t-1} \left(a(w_i) b(w_i) \right), \nonumber \\
\Rightarrow \hat{N}_{sat} &=& \alpha \beta - \sum_{i=0}^{t-1} a(w_i) b(w_i).
\end{eqnarray}
Considering our setting for $a(u)$ and $b(u)$ in the line \ref{alglin:Set_au_bu} of Algorithm \ref{Alg:SubOptimalInstance} we have
\beq
\label{eq:alpha_ui_alpha_ui1}
a(u_{i+1}) = a(u_i) + a(w_i), ~~ b(u_{i+1}) = b(u_i) + b(w_i),
\eeq
for $0 \leq i \leq t-1$ and either $(a(u_i), b(u_i)) = \left(\lceil a(u_{i+1})/2 \rceil, \lfloor b(u_{i+1})/2 \rfloor \right)$ or $(a(u_i), b(u_i)) = \left(\lfloor a(u_{i+1})/2 \rfloor, \lceil b(u_{i+1})/2 \rceil \right)$. In any case using eq. (\ref{eq:alpha_ui_alpha_ui1}) we have
\beqa
a(u_i) &\leq & \lceil a(u_{i+1})/2 \rceil, \\
&\leq & \frac{a(u_{i+1}) +1}{2}, \\
& = & \frac{a(u_{i}) + a(w_i) +1}{2},\\
\Rightarrow a(u_i) &\leq& a(w_i) + 1.
\eeqa
By a similar argument we have $b(u_i) \leq b(w_i)+1$. Using eq. (\ref{eq:alpha_ui_alpha_ui1}) recursively, it is easy to see that $\alpha = a(u_0)+ \sum_{i=0}^{t-1} a(w_i)$ and $\beta = b(u_0)+ \sum_{i=0}^{t-1} b(w_i)$.
Therefore, using eq. (\ref{eq:Nhat_satAlphBet}) and (\ref{eq:Nhat_satGamma}),
\beqa
\hat{N}_{sat} &=& \alpha \beta - \sum_{i=0}^{t-1} a(w_i) b(w_i), \\
&=& \sum_{i=0}^{t-1} \left( a(u_i) b(w_i) + a(w_i) b(u_i) \right), \\
&\leq & \sum_{i=0}^{t-1} \left( [a(w_i)+1] b(w_i) + a(w_i) [b(w_i)+1] \right), \\
&\leq & \sum_{i=0}^{t-1} \left( 2a(w_i)b(w_i) + a(w_i) + b(w_i) \right), \\
&\leq & \alpha + \beta +  2\sum_{i=0}^{t-1} a(w_i)b(w_i) , \\
\Rightarrow \sum_{i=0}^{t-1} a(w_i)b(w_i) &\geq& \frac{\alpha \beta - \alpha - \beta}{3}.
\eeqa
Finally, using the above inequality and eq. (\ref{eq:Nhat_satAlphBet}), we have
\beqa
N_{sat}(\alpha, \beta, K) &\leq & \hat{N}_{sat}, \\
&=& \alpha \beta - \sum_{i=0}^{t-1} \alpha(w_i) \beta(w_i), \\
&\leq &  \alpha \beta - \frac{\alpha \beta - \alpha - \beta}{3} = \frac{2\alpha \beta + \alpha + \beta}{3}.
\eeqa
Furthermore as $N_{sat}(\alpha, \beta, K)$ is an integer we conclude that
\begin{align*}
N_{sat}(\alpha, \beta, K) & \leq \left \lfloor \frac{2\alpha \beta + \alpha + \beta}{3} \right \rfloor.
\end{align*}
\end{proof}
The aforementioned proposed upper bound on the saturation number is tight. To see this, let consider $\beta = 1$. It is easy to see that $N_{sat}(\alpha, 1, K) = \alpha$ and using Lemma \ref{lem:UpBoundNsat} we have $N_{sat} \leq \lfloor \alpha + 1/3 \rfloor = \alpha$.

\begin{figure}[!t]
\centering
%
%
%
\includegraphics[scale=0.5]{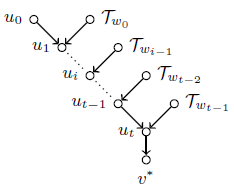}
\caption{{\small Saturation path}}
\label{Fig:SaturationPath}
\end{figure}


\section{Multiplicative Gap between upper and lower bounds}
\label{sec:mult_gap}
We now show that for any set of problem parameters, our proposed lower bound and the achievable rate of \cite{maddahN14} in eq. (\ref{eq:coded_rate}) are within a factor of four, i.e., we show the following result.
\begin{theorem}
\label{thm:mult_gap}
Consider a coded caching system with $N$ files and $K$ users each with a normalized cache size $M$. Then,
\beqa
\gamma(M) = \frac{R_c(M)}{R^\star(M)} \leq 4,
\eeqa
for $0 \leq M \leq N$.
\end{theorem}

The key idea in proving this result is to exploit the analytical upper bound on the saturation number $N_{sat}(\alpha, \beta, K)$ proposed in Section \ref{sec:subsec_anlytic_bound_Nsat}. For a given $N$ and $K$, we consider three distinct regions of $M$. For each range, an appropriate $(\alpha, \beta)$ pair allows us obtain a lower bound on the rate that is within a factor of four of the achievable rate. 

%
\begin{proof}

We use Corollary \ref{corollary:Lvalue} with the $2\alpha$ and $2\beta$, so that $P'_l$ and $P'_r$ have parameters $\alpha$ and $\beta$.
 This gives us the following lower bound.
\beqa
2\alpha R^\star(M) + 2\beta M \geq \min\left(2\alpha\min(2\beta,K), 2\alpha \beta + [N-N_0]^{+} \right),
\eeqa
Moreover, we restrict $2\beta \leq K$ so that,
\begin{align}
2\alpha R^\star(M) + 2\beta M &\geq \min\left(4\alpha\beta, 2\alpha \beta +N-N_0 \right) \nonumber\\
\implies R^\star(M) &\geq \min\left(2\beta, \beta + \frac{N-N_0}{2\alpha} \right) - \frac{\beta}{\alpha}M. \label{eq:ProofGapMainbound}
\end{align}

Our first observation is that for $\min(N,K) \leq 4$, the bound is easily seen to be true. Towards this end, by setting $\alpha = N, \beta = 1$ in (\ref{eq:ProofGapMainbound}), we obtain
\begin{align*}
R^\star(M) \geq 1 - \frac{M}{N}.
\end{align*}
where we used $N_{sat}(N, 1, K) = N$. Furthermore, from eq. (\ref{eq:coded_rate}),
\beqa
R_c(M) \leq \min(N,K) \left(1-M/N\right),
\eeqa
This means that $\gamma(M) = \min(N,K) \leq 4$ for $\min(N,K) \leq 4$.

Thus, in the subsequent discussion, we only consider $\min(N,K) \geq 5$.
As in \cite{maddahN14}, we divide the $M$-axis to three separated regions. For given $M$, we explore the space of $(\alpha, \beta)$ pairs to obtain an appropriate lower bound that allows us to show the multiplicative gap of four. 
\subsection{Region I: $0 \leq M \leq \max(1,N/K)$}
\label{sec:RegionI}
First, we consider the range $0 \leq M \leq 1$. In eq. (\ref{eq:ProofGapMainbound}) we set $\alpha = 1,\beta = \lfloor \min(N,K)/2 \rfloor$. By such a setting we have $2\beta \leq \min(N,K) \leq K$ and $N \geq N_{sat}(1,\beta, K) = \beta$. Therefore for $M \leq 1$,
\begin{align*}
R^\star(M) &\geq \min \left(2\beta, \frac{N+\beta}{2}\right) - \beta M\\
&\stackrel{(a)}{\geq}  \min \left(\beta, \frac{N-\beta}{2}\right)\\
&\stackrel{(b)}{\geq}  \min \left(\frac{\min(N,K)-1}{2}, \frac{N-\min(N,K)/2}{2}\right)\\
&\stackrel{(c)}{\geq}  \min \left(\frac{\min(N,K)-1}{2}, \frac{\min(N,K)}{4}\right)\\
&\stackrel{(d)}{\geq}  \frac{\min(N,K)}{4}\\
&\geq \frac{\min(N,K)(1 - M/N)}{4}\\
&\geq R_c(M)/4.
\end{align*}
Here, $(a)$ holds since $M \leq 1$, $(b)$ holds since $(\min(N,K)-1)/2 \leq \beta \leq \min(N,K)/2$, $(c)$ holds since $N \geq \min(N,K)$, and $(d)$ holds since $\min(N,K) \geq 2$.

Next, consider the range $M \in [1, N/K]$. Note that we only need to consider the scenario where $N \geq K$. The achievable rate $R_c(M)$ in this interval is upper bounded by the convex combination of the rates $R_c(0)$ and $R_c(N/K)$ so that
$$R_c(M) \leq \lambda R_c(N/K) + (1-\lambda)R_c(0) = K(1-\lambda/2)-\lambda/2,$$
where $\lambda = KM/N$. Now, we set $\alpha = \lceil N/K \rceil,\beta = \lfloor K/2 \rfloor$ so that $\alpha \beta \leq (N/K+1)K/2 = N/2+K/2 \leq N$. As, $N_{sat}(\alpha,\beta,K) \leq \alpha \beta$, this means that $N \geq N_{sat}(\alpha, \beta,K)$. In addition, note that $2 \beta \leq K$. Therefore, we can use eq. (\ref{eq:ProofGapMainbound}) to obtain
\begin{align}
\label{eq:proof_mult_gap_reg_I_p2}
R^\star(M) &\geq  \min\left\{2\beta,~ \beta + \frac{N - N_{sat}(\alpha, \beta,K)}{2\alpha} \right\} - \frac{\beta}{\alpha}M, \nonumber \\
&\stackrel{(a)}{\geq}  \min\left\{2\beta \left(1-\frac{M}{2\alpha}\right),~ \frac{2\beta}{3} + \frac{N - 2\beta M}{2\alpha} - \frac{\beta}{6\alpha} - \frac{1}{6} \right\} , \nonumber \\
&\stackrel{(b)}{\geq}  \min\left\{(K-1) \left(1-\frac{KM}{2N}\right),~ \frac{\beta}{2} + \frac{N - 2\beta M}{4N/K} -  \frac{1}{6} \right\} , \nonumber \\
& \geq   \min\left\{\frac{K}{2} \left(1-\frac{\lambda}{2}\right),~ \frac{\beta}{2}\left(1-\lambda \right) + \frac{K}{4} -  \frac{1}{6} \right\} , \nonumber \\
& \stackrel{(c)}{\geq} \min\left\{\frac{R_c(M)}{2},~ \frac{K}{2}\left(1-\frac{\lambda}{2}\right) -\frac{(1-\lambda)}{4}-  \frac{1}{6} \right\} , \nonumber \\
&\stackrel{(d)}{\geq}  \min \left\{ \frac{R_c(M)}{2},~ \frac{R_c(M)}{4} + \frac{(K-3)}{4}\left(1-\frac{\lambda}{2}\right) +  \frac{1}{3} \right\} , \nonumber\\
&\stackrel{(e)}\geq  R_c(M)/4,
\end{align}
where in $(a)$ we used Lemma \ref{lem:UpBoundNsat} to bound $N_{sat}(\alpha, \beta,K$), in $(b)$ we used $N-2\beta M \geq 0$, $1 \leq \alpha \leq N/K+1\leq 2N/K$, $(K-1)/2 \leq \beta$, and in $(c)$ we used $\beta \geq (K-1)/2$, $\lambda = KM/N$ and the expression for the upper bound on $R_c(M)$ above. Next, $(d)$ holds because of the achievable rate bound and $(e)$ holds since $\min(N,K) \geq 5$. Therefore, $\gamma(M) \leq 4$ for $M \in [1, N/K]$ and $N \geq K$. Thus, we conclude that we have $\gamma(M) \leq 4$ for $M \in [0, \max(1,N/K)]$.
\subsection{Region II: $\max(1,N/K) < M \leq N/2$}
\label{sec:RegionII}
For any $M \in [\max(N/K,1), N/2]$ we define $t_0 = \lfloor KM/N \rfloor$ so that $t_0N/K \leq M \leq (t_0+1)N/K$. Since $M \geq N/K$ thus $t_0 \geq 1$. Using eq. (\ref{eq:coded_rate}), it turns out that,
\begin{align*}
R_c(M) &\leq  R_c(t_0N/K),\\
&= \frac{K}{t_0+1} - \frac{t_0}{t_0+1}, \\
&\stackrel{(a)}{\leq}  \frac{K}{KM/N} - \frac{1}{2}, \\
&= \frac{N}{M} - \frac{1}{2},
\end{align*}
where $(a)$ holds since $t_0+1 \geq KM/N$ and $t_0 \geq 1$.


Now, consider setting $\alpha =\lfloor 2M \rfloor$ and $\beta = \lfloor N/2M \rfloor$. With this setting we have $\alpha \geq 2$ (since $M \geq 1$), $\beta \geq 1$ (since $M \leq N/2$), and $\beta \leq N/2M \leq K/2$ (since $M \geq N/K$). Furthermore, since $\alpha \beta \leq 2M \times N/2M =N$ and $N_{sat}(\alpha, \beta, K) \leq \alpha \beta$ therefore $N \geq N_{sat}(\alpha, \beta,K)$. This together with $2\beta \leq K$ implies that such a setting is a valid setting to use (\ref{eq:ProofGapMainbound}). Therefore, using Lemma \ref{lem:UpBoundNsat} to bound $N_{sat}(\alpha, \beta,K)$, we have
\beqa
R^\star(M) &\geq & \min\{2\beta, \frac{2\beta}{3} + \frac{N}{2\alpha} - \frac{\beta}{6\alpha} - \frac{1}{6} \} - \frac{\beta}{\alpha}M.
\eeqa

We claim that $2\beta \geq 2\beta/3 + N/2\alpha - \beta/6\alpha - 1/6$ or equivalently $8\alpha \beta + \alpha + \beta \geq 3N$. This can be seen as follows. When, $N/4 < M \leq N/2$ we have $\alpha > N/2, \beta = 1$, so that this holds. On the other hand when $\max(1, N/K) < M \leq N/4$, we have $\alpha \geq 2M -1$, $\beta \geq N/2M - 1$, so that $8\alpha\beta + \alpha + \beta \geq 8N - 7(N/2M+2M) + 6$. It can been seen that $N/2M + 2M \leq N/2+2$ for $1 \leq M \leq N/4$ therefore $8\alpha\beta + \alpha + \beta \geq 9N/2 -8 \geq 3N$ for $N\geq 6$. For $N = 5$, the claim trivially holds since $\alpha \geq 2, \beta \geq 1$ so that $8\alpha\beta + \alpha +\beta \geq 19 \geq 3\times N = 15$.

Thus, we have
\beqa
R^\star(M) &\geq & \frac{2\beta}{3} + \frac{N-2\beta M}{2\alpha} - \frac{\beta}{6\alpha} - \frac{1}{6} , \\
& \stackrel{(a)}{\geq} & \frac{7\beta}{12} + \frac{N-2\beta M}{4M} - \frac{1}{6}, \\
&=& \frac{N}{4M} + \frac{\beta}{12} - \frac{1}{6}, \\
& \stackrel{(b)}{\geq} & \frac{N}{4M}  - \frac{1}{12}, \\
&\geq & \frac{N}{4M}  - \frac{1}{8}\\
& \geq & \frac{R_c(M)}{4},
\eeqa
where in $(a)$ we used $N - 2\beta M \geq 0$, $\alpha \geq 2$ and $\alpha \leq 2M$ and in $(b)$ we used $\beta \geq 1$. Eventually, $\gamma(M) \leq 4$ for $\max(N/K,1) \leq M \leq N/2$.
\subsection{Region III: $N/2 < M \leq N$}
\label{sec:RegionIII}
Let $t_0 = \lfloor K/2 \rfloor$ so that $M \geq t_0N/K$ for $M \in (N/2, N]$. For any $M \in (N/2,N]$ the convex combination of rate $R_c(t_0N/K)$ and $R_c(N)$ gives us $R_c(M) \leq \lambda R_c(t_0 N/K) + (1-\lambda)R_c(N) = \lambda R_c(t_0N/K)$ where $M = \lambda t_0 N/K + (1-\lambda)N$ or equivalently $\lambda = (1-M/N)/(1-t_0/K)$. According to this and eq. (\ref{eq:coded_rate}) we observe that,
\begin{align*}
R_c(M) &\leq \lambda R_c(t_0N/K),\\
&= \frac{(1-M/N)}{(1-t_0/K)} \frac{(K-t_0)}{(t_0+1)}, \\
&= \frac{K(1-M/N)}{(1+t_0)}, \\
&\stackrel{(a)}{\leq}  \frac{K(1-M/N)}{K/2}, \\
&= 2(1-M/N),
\end{align*}
where $(a)$ holds since $1+t_0 = 1 + \lfloor K/2 \rfloor \geq K/2$.

Now if we set $\alpha = N$ and $\beta = 1$ in (\ref{eq:ProofGapMainbound}) we obtain
\beqa
R^\star(M) &\geq & 1 - M/N \\
&\geq & \frac{R_c(M)}{2}.
\eeqa
This implies that $\gamma(M) \leq 2 \leq 4$ for $M \in [N/2, N]$ and concludes the proof.
\end{proof}

\section{Lower bounds on the other variants of the coded caching problem}
\label{sec:variants}
In addition to the original coded caching problem there are many variants of the problem including coded caching with multiple requests \cite{jiTLC14}, decentralized coded caching \cite{maddahN14mr_tradeoff} and caching in device to device wireless networks \cite{jiCM13}. Our proposed strategy applies with minor changes for these problems. 

\subsection{Caching in device to device wireless networks}
Wireless device to device (D2D) networks where communication is limited to be single-hop are studied in \cite{jiCM13}. There are $K$ users who are the nodes of the network. Each user has a cache of size $M$ and $N$ files are stored across the different user caches. Thus, in this setting we necessarily have $KM \geq N$. As in the coded caching problem there are placement and delivery phases. In the placement phase the caches are populated from a server; this phase does not depend on the user demands. The server then leaves the network. We let $Z_i$ represent the cache content of the $i$-th user. In the delivery phase each user requests a file and the remaining users are informed about this request. Based on the requests, each user broadcasts a signal so that all demands can be satisfied. We denote by $X^{(i)}_{d_1,\ldots,d_K}$ the signal that is broadcasted in the delivery phase by the $i$-th user when the $j$-th user requests file $d_j \in [N]$ for $1\leq j \leq K$. The delivery signal sent by each user is function of its cache content so that $H(X^{(i)}_{d_1,\ldots,d_K}|Z_i)=0$. We also denote by $X_{d_1,\ldots,d_K}$ the set of signals sent by all the users, i.e., $X_{d_1,\ldots,d_K} = \{X^{(1)}_{d_1,\ldots,d_K},\ldots,X^{(K)}_{d_1,\ldots,d_K} \}$. The rate of the signal that the $i$-th user sends in the delivery phase is denoted by $R_{i,d_1,\ldots,d_K}(M)$. We are interested in lower bounding the worst case rate that denoted by $R^\star(M)=K\max_{i,d_1,\ldots,d_k} R_{i,d_1,\ldots,d_K}(M)$.

The cut-set technique and Han's inequality have been studied in \cite{jiCM13} and \cite{senguptabeyondd2d} respectively to establish lower bound on $R^\star(M)$. The multiplicative gap established in \cite{jiCM13} depends on $M$ and is not constant, whereas \cite{senguptabeyondd2d} shows a gap of at most $8$.

The D2D setting is almost exactly the same as the coded caching setting studied in our work. Our technique for obtaining lower bounds is applicable here with essentially no change and we can use Theorem \ref{thm:Lvalue} and its corollary. Furthermore, since $H(X^{(i)}_{d_1,\ldots,d_K}|Z_i)=0$ we can get lower bounds that are somewhat tighter. By treating $X_{d_1,\ldots,d_K}$ as the delivery signal of the original coded caching problem, we can our lower bound to show that the multiplicative gap between the achievable rate in \cite{jiCM13} and our proposed lower bounds is at most $4$. The proof is quite similar to that of Theorem \ref{thm:mult_gap} and is omitted.

\subsection{Coded caching with multiple requests}
Coded caching with multiple requests is variation of the original problem in which each user requests $l$ files from the server in the delivery phase. A straightforward achievable scheme in this setting is to apply the scheme of \cite{maddahN14} $l$ times. This problem is investigated in \cite{jiTLC14} where a new achievable scheme is proposed based on multiple groupcast index coding. Furthermore, \cite{jiTLC14} introduce a cut-set type lower bound and show that their scheme is within a multiplicative factor of $18$ to the lower bound. 
In contrast, using our approach we can demonstrate a multiplicative gap of $4$ for this problem as well.


In this setting the only difference with respect to the original problem is that from a cache signal $Z_i$ and delivery signal $X_{d_1,\ldots,d_K}$ one can recover up to $l$ distinct files.
Thus, $d_i$ is a vector of size $l$ containing information about the $l$ files requested by $i$-th user. Therefore, all statements we presented for the original problem are applicable here, bearing in mind that $Rec(Z_i,X_{d_1,\ldots,d_K})$ can be as large as $l$. 
For instance, an extension of eq. (\ref{eq:BoundOnPsi}) gives us $L \leq l \alpha \min(\beta,K)$.
Similarly, the saturation number $N_{sat}(\alpha,\beta,K,l)$ is defined as the minimum $N'$ among all problem instance $P(\calT, \alpha,\beta,L,N',K,l)$ so that $L=l\alpha \min(K,\beta)$. It is easy to verify that $N_{sat}(\alpha,\beta,K,l) \leq l \alpha \min(\beta,K)$ in a similar way. The following claim can be shown (we omit the proof as it very similar to the previous discussion).
\begin{claim}
\label{clm:LwrBndMultRqst}
Consider a coded caching system with a server containing $N$ files and $K$ users. Each user has a cache of size $M$ and demands $l$ files in the delivery phase. The following lower bound holds for $N \geq N_0$ where $N_0 = N_{sat}(\alpha, \beta, K,l)$,
\beqa
\alpha R^\star(M) + \beta M \geq \min\left(2l\alpha\min(\beta,K),~ l\alpha\min(\beta,K)+(N-N_0)/2) \right).
\eeqa
\end{claim}
Similarly, an extension of the Lemma \ref{lem:UpBoundNsat} holds so that $N_{sat}(\alpha,\beta,K,l) \leq l(2\alpha \beta + \alpha +\beta)/3$ for $\beta \leq K$. Exploiting this upper bound and Claim \ref{clm:LwrBndMultRqst}, we are able to show that the multiplicative gap of the straightforward achievable scheme and our lower bound is at most $4$. Let $R^l_c(M) = l R_c(M)$ where $R_c(M)$ is defined in eq. (\ref{eq:coded_rate}).
\begin{theorem}
\label{thm:mult_gap_mult_req}
Consider a coded caching system with a server containing $N$ files and $K$ users. Each user requests $l$ files, and has a cache of size $0 \leq M \leq N$. Then
\beqa
\frac{R^l_c(M)}{R^\star(M)} \leq 4.
\eeqa
\end{theorem}

\begin{proof}
We divide the $M$ axis into three regions, $0 \leq M \leq \max(l,N/K)$, $\max(l,N/K) \leq M \leq N/2$, and $N/2 \leq M \leq N$. In each region we show $R^l_c(M)/R^\star(M) \leq 4$ for any $N$ and $K$. In the following proof, $M=l$ plays the same role as $M=1$ in proof of Theorem \ref{thm:mult_gap}. Before embarking on the proof, we note that we only need to analyze the gap for $\min(N,lK) \geq 5$. Note that the lower bounds of the original problem are also valid here. Indeed, if each user instead of requesting $l$ distinct files request the same file $l$ times then the problem will be equivalent to the original one. Now, in (\ref{eq:ProofGapMainbound}) if we set $\alpha=N$ and $\beta=1$ then we get $NR^\star+M\geq N$, or equivalently $R^\star(M)\geq (1-M/N)$, which is applicable to the multiple request problem. Regarding that $R^l_c(M)\leq \min(N,lK)(1-M/N)$, therefore $R^l_c(M) / R^\star(M) \leq 4$ for $(N,lK)\leq 4$.
\subsubsection{Region I: $0 \leq M \leq \max(l,N/K)$}
For $0 \leq M \leq \max(l,N/K)$, we first show that the result holds for $M \leq l$. Since we separately analyze the gap for $M \geq N/2$ we assume $l \leq N/2$ so that $M \leq \max(l,N/K) \leq N/2$. We use result of the Claim \ref{clm:LwrBndMultRqst} with setting $\alpha=1$ and $\beta=\lfloor \min(N/2l,K/2) \rfloor$ where $\beta \geq 1$ from $l \leq N/2$. Following the exact same steps as in Section \ref{sec:RegionI} for $M\leq 1$, it turns out that $R^\star(M) \geq  \min(N,lK)/4 \geq R^l_c(M)/4$ for $M\leq l$.

Now, we assume that $l \leq M \leq \max(l,N/K)$ which is nonempty if $N/K\geq l$. Therefore, we only need to analyze the gap for $N \geq lK$ and $l \leq M \leq N/K$. In this range of $M$ the convex combination of $M=0$ and $M=N/K$ is achievable so that $R^l_c(M) \leq \lambda R^l_c(N/K)+(1-\lambda)R^l_c(0)$. From $R^l_c(0)=lK$ and $R^l_c(N/K)=l(K-1)/2$ we have $R^l_c(M) \leq lK(1-\lambda/2)-l\lambda/2$ where $\lambda = KM/N$. By setting $\alpha = \lceil N/lK \rceil$ and $\beta = \lfloor K/2 \rfloor$, we have $\alpha \beta \leq \alpha K/2 \leq N/2l+K/2 \leq N/l$ (from $lK \leq N$) and that $N_{sat}(\alpha,\beta,K,l)\leq l\alpha\beta \leq N$. This ensures that the setting is valid for using Claim \ref{clm:LwrBndMultRqst}. According to Claim \ref{clm:LwrBndMultRqst} for such a setting we have,
\beqa
R^*(M) &\geq &  \min\left(2l \beta, l \beta + \frac{N- N_{sat}(\alpha, \beta, K, l)}{2\alpha} \right) - \frac{\beta M }{\alpha},\\
&\stackrel{(a)}{\geq} &  \min\left(\frac{l K}{2}\left( 1- \frac{\lambda }{2} \right),  \frac{l K(1-\lambda/2)}{2} - \frac{l (1-\lambda)}{4} - \frac{l}{6} \right) ,\\
&\stackrel{(b)}{\geq} &  \min\left(\frac{R_c(M)}{2},  \frac{l K(1-\lambda/2)}{4} +\frac{l (1-\lambda/2)}{2}- \frac{l(1-\lambda)}{4}-\frac{l}{6} \right) ,\\
&=&  \min\left(\frac{R_c(M)}{2},  \frac{l K(1-\lambda/2)}{4} +\frac{l}{12} \right) ,\\
&\geq &  \min\left(\frac{R_c(M)}{2},  \frac{R_c(M)}{4}  \right) \geq \frac{R_c(M)}{4},
\eeqa
where inequality (a) can be obtained by making the same argument as we made in first five lines of eq. (\ref{eq:proof_mult_gap_reg_I_p2}) and (b) from $K \geq 2$.
\subsubsection{Region II: $\max(l,N/K)\leq M \leq N/2$}
In the first step, we try to get an upper bound on the achievable rate. Letting $t_0=\lfloor KM/N \rfloor$ and following the argument we made in Section \ref{sec:RegionII} gives us $R^l_c(M) \leq lR_c(M) \leq l\left(N/M-1/2\right)$ for $M$ in this range. Next, by setting $\alpha =\lfloor 2M/l \rfloor$ and $\beta = \lfloor N/2M \rfloor$ we have $ N_{sat}(\alpha,\beta,K,l)\leq l\alpha \beta \leq N$ and $\beta \leq 2N/M \leq K/2$ by $M\geq N/K$ which imply that the constraints of the Claim \ref{clm:LwrBndMultRqst} are satisfied. Therefore,
\beqa
R^\star &\geq & \min\left(2l\beta, ~ l\beta+\frac{N-N_{sat}(\alpha,\beta,K,l)}{2\alpha} \right)- \frac{\beta M}{\alpha},\\
&\stackrel{(a)}{\geq} & \min\left(2l\beta\left(1-\frac{M}{2l\alpha} \right), ~ \frac{7l\beta}{12}+\frac{N-2\beta M}{2\alpha} - \frac{l}{6} \right) ,\\
&\stackrel{(b)}{\geq} & \min\left(2l\beta\left(1-\frac{M}{2M} \right), ~ \frac{7l\beta}{12} +\frac{N-2\beta M}{4M/l} - \frac{l}{6}\right) ,\\
&\stackrel{(c)}{\geq} & \min\left(\frac{Nl}{4M}, ~ \frac{Nl}{4M}- \frac{l}{12}\right) ,\\
&\geq & R^l_c(M)/4 ,
\eeqa
where in (a) we used upper bound on $N_{sat}(\alpha,\beta,K,l)$ and that $\beta / \alpha \leq \beta /2$ (from $\alpha \geq 2$), in (b) we used $N-2\beta M \geq 0$, $\alpha \leq 2M/l $, and $\alpha \geq 2M/l-1 \geq M/l$ (from $M \leq l$). In (c) we used $\beta \geq K/4$ (for $K\geq 2$) and $\beta \geq 1$ (from $M \leq N/2$).
\subsubsection{Region III: $N/2 \leq M \leq N$}
Using the same argument we made in Section \ref{sec:RegionIII} the achievable rate is bounded by $R^l_c(M) \leq lR_c(M) \leq 2l\left(1-M/N\right)$. According to Claim \ref{clm:LwrBndMultRqst} by setting $\alpha =\lfloor N/l \rfloor$ and $\beta=1$ one may not recover all $N$ files since $\alpha l \leq N$, but if we increase $\alpha$ to $\lceil N/l \rceil$ then all files will be recovered. Therefore $\alpha R^\star(M) + M \geq N$ or equivalently $R^\star(M) \geq (N-M)/\alpha$. From $N-M \geq 0$ and that $\alpha \leq N/l+1\leq 2N/l$ (since $l \leq N$) it turns out that $R^\star(M) \geq l(1-M/N)/2 \geq 4R^l_c(M)$ for $N/2 \leq M \leq N$. This concludes the proof.
\end{proof}

\subsection{Decentralized coded caching}
In the original coded caching problem the placement phase is managed by a central server. However, in many scenarios such coordinated placement phase may be impractical. Instead, a decentralized placement phase was investigated in \cite{maddahN14mr_tradeoff} where the users cache random subsets of the bits of each file while respecting the cache size constraint. Even in this setting a multiplicative gap of $12$ to the cut-set lower bound was obtained. Note that the lower bounds established for the centralized coded caching problem are also applicable to the decentralized case. By similar techniques to those used in proof of Theorem \ref{thm:mult_gap} we can establish a multiplicative gap of $4$. The proof is omitted as it is quite similar.


\section{Comparison with existing results}
\label{sec:comparison}
Lower bounds on the coding caching rate have been proposed in independent work as well. In this section we compare our lower bounds with other approaches.
\subsection{Comparison with cutset bound}
Our first observation is that the cutset bound in \cite{maddahN14} is a special case of the bound in eq. (\ref{eq:cap_lower_bound}). In particular, suppose that $\alpha = \lfloor N / s \rfloor$, $\beta = s$ for $s = 1,\ldots,\min(N,K)$. In this case, we have $\alpha\beta \leq N$. Thus, it is easy to construct a problem instance where $L = \alpha \beta$ (see Corollary \ref{corollay:sat_rem}). This also follows from observing that $N_{sat}(\alpha,\beta,K) \leq \alpha \beta$.

Our bound allows us to explore a larger range of $(\alpha, \beta)$ pairs that in turn lead to better lower bounds on $R^{\star}$. Suppose that for a coded caching system with $N$ files and $K$ users, we first apply the cutset bound with certain $\alpha_1$ and $\beta_1$ such that $\alpha_1\beta_1 < N$. This would result in the inequality
\begin{align*}
\alpha_1 R^{\star} + \beta_1 M \geq \alpha_1 \beta_1.
\end{align*}
However, our approach can do strictly better. To see this note that $\alpha_1 \beta_1 < N$ implies that $N_{sat}(\alpha_1, \beta_1,K) < N$. Now, using Corollary \ref{corollary:Lvalue} we can instead attempt to lower bound $2\alpha_1 R^{\star}+ 2\beta_1 M$ and obtain the following inequality. 
\begin{align*}
2 \alpha_1 R^{\star} + 2 \beta_1 M  &\geq \min\left(4\alpha_1\beta_1,~2 \alpha_1\beta_1 + N - N_{sat}(\alpha_1, \beta_1, K) \right)\\
\implies \alpha_1 R^{\star} + \beta_1 M &\geq \min\left(2\alpha_1\beta_1,~\alpha_1\beta_1 + (N - N_{sat}(\alpha_1, \beta_1, K))/2 \right),
\end{align*}
which is strictly better than the cutset bound since $N - N_{sat}(\alpha_1, \beta_1, K) > 0$.
\begin{example}\label{eg:discuss_1}
Consider a system containing a server with four files and three users, $N=4$ and $K=3$. The cutset bounds corresponding to the given system are
\begin{align*}
4R^{\star} + M &\geq 4,\\
2R^{\star}+2M &\geq 4, \text{~and}\\
R^{\star}+3M &\geq 3.
\end{align*}
A simple calculation shows that if $M=1$, the above inequalities, yield the lower bound $R^{\star} \geq 1$.

Now, consider the second bound, $2R^{\star}+2M \geq 4$ and instead attempt to obtain a lower bound on $4R^{\star}+4M$.
In this case by exhaustive enumeration, it can be verified that $N_{sat}(2,2,3) = 3 < N$. 
Using Corollary \ref{corollary:Lvalue}, this results in the lower bound $L^* \geq \min(4\times 3, 2\times4+4-N_{sat}(2,2,3))=9$. Thus we can conclude $R^{\star}+M \geq 2.25$ which is better than the cutset bound $R^{\star}+M \geq 2$. Moreover, this inequality also yields a better lower bound $R^{\star} \geq 1.25$.
\end{example}


\subsection{Comparison with lower bound of \cite{sengupta2015improved}}
The authors in \cite{sengupta2015improved} use Han's inequality \cite[Theorem 17.6.1]{cover2012elements} to establish the following lower bounds on the coded caching problem.
\beq
\label{eq:AvikBound}
\alpha R^\star(M) + \beta M \geq N - \frac{\mu}{\mu + \beta}[N-\alpha \beta]^+ - [N-\alpha K]^+,
\eeq
where $\mu = \min(\lceil \frac{N-\alpha \beta}{\alpha} \rceil, K- \beta)$, $\beta \in \{1,\ldots,K \}$ and $\alpha \in \{1, \ldots, \lceil \frac{N}{\beta} \rceil\}$. This bound also provides more flexibility in the choice of $\alpha$ as compared to the cutset bound.

An analytical comparison between our bound and the bound in inequality (\ref{eq:AvikBound}) is hard, especially since a priori in all these bounds, for a given $M$, it is unclear which particular $(\alpha,\beta)$ pair gives the best lower bound. Thus, in the discussion below we attempt to analytically compare the bounds for given $(\alpha, \beta)$. We also present a numerical comparison in Section \ref{sec:numerical_comp}. 
\begin{itemize}
\item[(a)] Our bound is superior, when $ 1/\alpha + 1/\beta \leq 0.4$, i.e., when the values of $\alpha$ and $\beta$ are large enough. Note that the best lower bounds on $R^{\star}(M)$ for systems with $N$ and $K$ reasonably large are obtained for higher values of $\alpha$ and $\beta$. Thus, for most parameter ranges our bounds are better.
\item[(b)] The bound in \cite{sengupta2015improved} is better when $\alpha=1$ and $N \leq K$. This in turn means that their corresponding lower bound for small values of $M$ is better than ours.
\item[(c)] We can demonstrate that our proposed lower bound is within a factor of four of the achievable rate, whereas \cite{sengupta2015improved} only demonstrates a multiplicative gap of eight.
\end{itemize}
In the remainder of this discussion we assume that $\alpha \geq 2$ and show these claims. Let $L^*$ denote the value of our lower bound and let $L_{H}$ denote the lower bound of \cite{sengupta2015improved}.

\noindent {\it Case 1:} {\it $\alpha \beta > N$.}\\
Note that $\alpha \leq \lceil N/\beta \rceil$ in inequality (\ref{eq:AvikBound}). Furthermore, $\alpha \geq 2$ implies that $N \geq \beta$. Thus, we can conclude that $\alpha \beta \leq  \lceil N/\beta \rceil \beta \leq 2N$. Now, we use Corollary \ref{corollary:Lvalue} to compare the bounds. Specifically, set $\alpha_l = \lceil \alpha/2 \rceil, \beta_l = \lfloor \beta/2 \rfloor, \alpha_r = \lfloor \alpha/2 \rfloor$ and $\beta_r = \lceil \beta/2 \rceil$. This implies that
\begin{align*}
\max (\alpha_l \beta_l,\alpha_r \beta_r) \leq \frac{\alpha \beta}{2} \leq N.
\end{align*}
Thus, we obtain $L^* = \min\left(\alpha \beta, ~ \alpha_l \beta_l+\alpha_r \beta_r+N-N_0 \right)$. Note that $$N_0 = \max\left(N_{sat}(\alpha_l,\beta_l,K),~N_{sat}(\alpha_r,\beta_r,K) \right) \leq \max(\alpha_l\beta_l, \alpha_r \beta_r) \leq N. \text{~~(from above)}$$
Thus,
\beqa
L^* &= & \min\{\alpha \beta,~ \alpha_l \beta_l + \alpha_r \beta_r + N-N_0 \}\\
&\geq & \min\{\alpha \beta,~ \alpha_l \beta_l + \alpha_r \beta_r + N-\max\left( \alpha_l \beta_l, ~\alpha_r \beta_r \right) \}\\
&= & \min\{\alpha \beta, \min\left( \alpha_l \beta_l, ~\alpha_r \beta_r \right) + N \}\\
&> & N.
\eeqa
On the other hand note that $L_{H}$ is at most $N$. Thus, our bound is strictly better.

\noindent {\it Case 2(a):} {\it $\alpha \beta \leq \alpha K \leq N$.}\\

As $N \geq \alpha \beta \geq N_{sat}(\alpha,\beta,K)$ we use (\ref{eq:ProofGapMainbound}) to obtain
\begin{align*}
L^* = \min\left(\alpha\min(K, 2\beta),~\alpha \beta + (N-N_0)/2 \right).
\end{align*}
The corresponding bound $L_{H}$ is obtained by setting $\mu = K-\beta$.
\begin{align*}
L_{H} &= \alpha K - (1-\beta/K)(N-\alpha\beta)\\
&= \alpha\beta(1+1/x-x)-(1-x)N, \text{~(where $0 \leq x=\beta/K \leq 1$)}\\
&\leq \alpha\beta(2-x), \text{~(since, $N \geq \alpha K =\alpha\beta/x$)}.
\end{align*}
Thus, we conclude that $L_{H} \leq \min (\alpha K, \alpha \beta(2-x)) \leq \alpha\min(K, 2\beta)$. As a result, we only need to examine whether $\alpha \beta + (N-N_0)/2 \geq L_H$. Now, using the fact that $N_0 \leq (2\alpha\beta+\alpha+\beta)/3$, we have that $L^* \geq L_{H}$ when
\begin{align}
2\alpha\beta/3+N/2-(\alpha+\beta)/6 &\geq \alpha\beta(1+1/x-x)-(1-x)N \nonumber\\
\implies (3/2-x)N-(1/x+1/3-x)\alpha \beta -(\alpha+\beta)/6 &\geq 0. \label{eq:intermed_comp_avik}
\end{align}
As $N \geq \alpha K = \alpha\beta/x$, inequality (\ref{eq:intermed_comp_avik}) certainly holds if
\begin{align*}
(1/2x+x-4/3)\alpha\beta - (\alpha+\beta)/6 \geq 0.
\end{align*}
It can be verified that $1/2x+x - 4/3 \geq \sqrt{2} -4/3 \geq 1/15$ for $0\leq x \leq 1$, so that the above inequality will definitely hold if $0.4 \geq 1/\alpha + 1/\beta$ which is the case for $\alpha,\beta \geq 5$.

{\it Case 2(b):} {\it $\alpha \beta \leq N < \alpha K$.}

In this case $\mu = \lceil N/\alpha - \beta \rceil$, so that
\begin{align*}
L_H &\leq N - (1 - \alpha\beta/N)(N - \alpha\beta)\\
&= \alpha\beta(2 - x') \text{~(where $0 \leq x' = \alpha\beta/N \leq 1$)}\\
\end{align*}
As in the previous case, we conclude that $L^* \geq L_H$ if
\begin{align*}
2\alpha\beta/3+N/2-(\alpha+\beta)/6 &\geq \alpha\beta(2-x').
\end{align*}
Upon analysis similar to the previous case, we can conclude that our bound is better when $0.4 \geq 1/\alpha + 1/\beta$.

\subsection{Comparison with lower bound of \cite{ajaykrishnan2015critical}}
The work of \cite{ajaykrishnan2015critical} is closest in spirit to our proposed lower bound. In particular, we show that their lower bound corresponds to specific problem instance as defined in our work. We note however that the work of \cite{ajaykrishnan2015critical} does not analyze the multiplicative gaps between the achievable rates and lower bounds. The lower bounds in \cite{ajaykrishnan2015critical} can be rewritten as

\begin{eqnarray}
\label{eq:TIFR_LowrBound}
2m R^\star + 2tm M \geq L_0, &~& \text{for } t \leq N,~ K\geq 2\\
2tm R^\star + 2m M \geq L_0,&~& \text{for } t \leq N,~ K\geq 2t,\nonumber
\end{eqnarray}
where $L_0=\min\{ 4tm^2,~2tm^2+N-\tilde{N_0}\}$, $\tilde{N}_0 = t(m^2-m+1)$, $m = n-\gamma$ and $n = \lceil (t+\sqrt{t^2+12t(N-t)})/6t \rceil$. Also, $\gamma=\max\left(0,~\lceil n- K/2t\rceil \right)$ and $\gamma=\max\left(0,~\lceil n- K/2\rceil \right)$ in the first and second lower bounds respectively.
We present these bounds using our notation so that $(\alpha,\beta)$ is equal to $(2m,2tm)$ and $(2tm,2m)$ in the first and second lower bounds in (\ref{eq:TIFR_LowrBound}) respectively. Note however, that in the above bound the only free parameter is $t$, i.e., $m$ itself is dependent on $t$. It is easy to see that $\beta \leq K$ therefore, unlike our method, this method cannot be used to obtain lower bounds when $\beta> K$.

\begin{figure}[!t]
\centering
\includegraphics[scale=0.5]{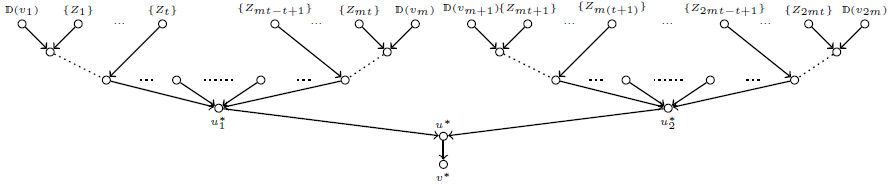}
\caption{{\small Problem instance associated with the lower bounds in \cite{ajaykrishnan2015critical}}}
\label{Fig:TIFR_PrblmInstnc}
\end{figure}

The lower bound $L_0$ in eq. (\ref{eq:TIFR_LowrBound}) above is reminiscent of our lower bound if the term $\tilde{N}_0$ is interpreted as a bound on the saturation number. In fact, for the specific setting of $(\alpha,\beta) = (m,mt)$, we can create a problem instance as described below, that is a saturated instance with exactly $t(m^2 - m+1)$ files, so that we can infer that $N_{sat}(m,tm,K) \leq t(m^2-m+1)$. It turns out that this upper bound on the saturation number may be slightly stronger than the one we derived in Lemma \ref{lem:UpBoundNsat} for general $\alpha$ and $\beta$ when $t$ and $m$ are small.
The associated problem instance of the first lower bound in (\ref{eq:TIFR_LowrBound}) is depicted in Fig. \ref{Fig:TIFR_PrblmInstnc}. The corresponding instance for the second lower bound in (\ref{eq:TIFR_LowrBound}) can be derived in a similar manner. In this figure, delivery phase signals $\dsD(v_1),\ldots, \dsD(v_{2m})$ are same as the delivery phase signals defined in \cite{ajaykrishnan2015critical}.
For this tree, it can be verified that the instance can be saturated with $t(m^2 - m+1)$ files, so that $N_{sat}(m,tm,K) \leq t(m^2-m+1)$.

However, an application of Algorithm \ref{Alg:SubOptimalInstance} will result in even better upper bound on the saturation number as shown in the example below. In particular, Algorithm \ref{Alg:SubOptimalInstance} will generate a different tree when trying to upper bound the saturation number.
\begin{example}
We consider a system with $N=64$ files and $K=8$ users and set $t=2$ in eq. (\ref{eq:TIFR_LowrBound}) so that $m=4$ and $\tilde{N}_0=26$. Algorithm \ref{Alg:SubOptimalInstance} for such a setting returns $N_{sat}(4,8,8)\leq 22$ which is smaller than $\tilde{N}_0$. This reduction in saturation number is a consequence of splitting $\alpha$ and $\beta$ equally in the Algorithm (\ref{Alg:SubOptimalInstance}) and continuing recursively thereafter. On the other hand, it can be noted that in Fig. \ref{Fig:TIFR_PrblmInstnc}, node $u^*_1$ is such that it has $m=4$ incoming edges which makes the corresponding lower bound looser ({\it cf.} Claim \ref{claim:incomingedgelimits}). 
\end{example}

\subsection{Comparison with results in \cite{tian2015note}}
In \cite{tian2015note} the author provides lower bounds for the specific case of $N=K=3$. The inequalities are generated via a computational technique that works with the entropic region of the associated random variables. Some of the bounds presented in \cite{tian2015note} can be obtained via our approach as well. However, the specific inequalities $3R^\star+6M\geq 8$, $18R^\star+12M\geq 29$ and $6R^\star+3M\geq 8$ cannot be obtained using our approach and strictly improves our region.
Note however, that it is not clear whether these inequalities can be obtained in a computationally tractable manner for the case of large $N$ and $K$.

\subsection{Numerical comparison of the various bounds}
\label{sec:numerical_comp}
We conclude this section, by providing numerical results for two cases: (i) $N=16, K=30$ and (ii) $N=64, K=50$. In Fig. \ref{Fig:comparisontoothrbounds} the ratio $R_c(M)/R^\star(M)$ is plotted by lower bounding $R^\star(M)$ by different methods. In case I (see Fig. \ref{Fig:comparisontoothrbounds}) we have $N=16$ and $K=30$. Our bound has the minimum multiplicative gap except in the small range $0\leq M\leq 1$. Specifically, as discussed previously, the bound in \cite{sengupta2015improved} is better than ours when $K \geq N$  and $\alpha=1$ and $0
\leq M \leq 1$. 
In case II, where $N > K$ our bound has minimum multiplicative gap for all range of $M$. 

\section{Conclusions and Future Work}
\label{sec:conclusions}
In this work we have considered a coded caching system with $N$ files, $K$ users each with a normalized cache of size $M$. We demonstrated an improved lower bound on the coded caching rate $R^\star(M)$. Our approach proceeds by establishing an equivalence between a sequence of information inequalities and a combinatorial labeling problem on a directed tree. Specifically, for given positive integers $\alpha$ and $\beta$, we generate an inequality of the form $\alpha R^\star + \beta M \geq L$. We showed that the {\it best} $L$ that can be obtained using our approach is closely tied to how efficiently a given number of files can be used by our proposed algorithm. Formalizing this notion, we studied certain structural properties of our algorithm that allow us to quantify the improvements that our approach affords. In particular, we show a multiplicative gap of four between our lower bound and the achievable rate. An interesting feature of our algorithm is that it is applicable for general value of $N, K$ and $M$ and is strictly better than all prior approaches for most parameter ranges.

There are still gaps between the currently known lower bounds and the achievable rate and an immediate open question is whether this gap can be reduced or closed. It would also be of interest to better understand coded caching rates in more general scenarios such as the hierarchical coded caching setup and for more general network topologies.

\definecolor{mygreen}{RGB}{0,205,205}
\definecolor{mygrey}{RGB}{131,139,139}
\begin{figure}[!t]
\centerline{
\includegraphics[scale=0.5]{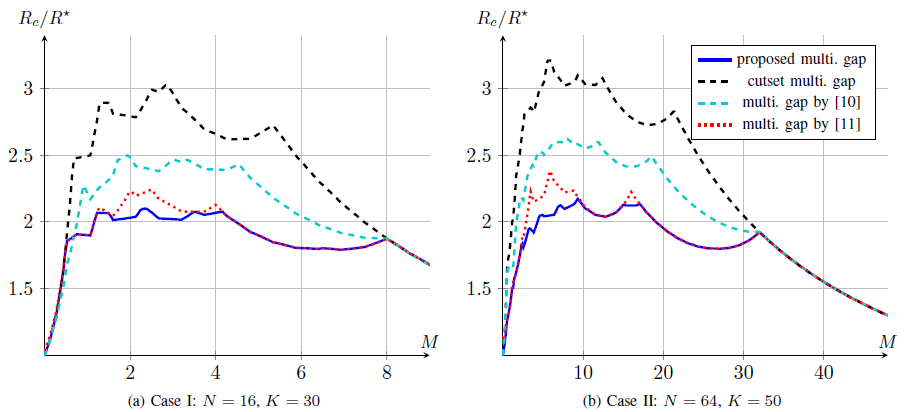}
}
\caption{{\small The plot demonstrates the multiplicative gap between the achievable rate, $R_c(M)$, in \cite{maddahN14} and lower bounds $R^\star(M)$ using different lower bounding techniques. For case II our lower bound results in the least multiplicative gap. In case I, where $N \leq K$, the multiplicative gap obtained by our proposed lower bound is lower than the others for $M \geq 1$. In the range $0 \leq M \leq 1$, \cite{sengupta2015improved} provides a slightly better result. }}
\label{Fig:comparisontoothrbounds}
\end{figure}

\bibliographystyle{IEEETran}
\bibliography{coded_caching,caching_refs}

\appendix
\begin{lemma}\label{lemma:general_lower_bd}
Algorithm \ref{Alg:Labeling} always provides a valid lower bound on $\alpha R^{\star} + \beta M$ where $\alpha = \sum_{i=1}^\ell |\dsD(v_i)|$ and $\beta = \sum_{i=1}^\ell |\dsZ(v_i)|$.
\end{lemma}
\begin{proof}
Consider any internal node $v \in \calT$. We have
\begin{align*}
&\sum_{u \in in(v)} H(\dsZ(u) \cup \dsD(u)|\dsW(u) \cup W_{new}(u)), \\
&\stackrel{(a)}{\geq}  \sum_{u \in in(v)} H( \dsZ(u) \cup \dsD(u)| \dsW(v)),\\
&\stackrel{(b)}{\geq}  H( \dsZ(v) \cup \dsD(v)| \dsW(v)), \\
&\stackrel{(c)}{=} I(W_{new}(v); \dsZ(v) \cup \dsD(v)| \dsW(v)) \\
&+ H(\dsZ(v) \cup \dsD(v) | \dsW(v) \cup W_{new}(v)),
\end{align*}
where inequality in $(a)$ holds since $\dsW(u) \cup W_{new}(u) \subseteq \dsW(v)$ and conditioning decreases entropy, $(b)$ holds since $\cup_{u \in in(v)} \dsZ(u) = \dsZ(v)$ and $\cup_{u \in in(v)} \dsD(u) = \dsD(v)$ and $(c)$ holds by the definition of mutual information. Let $V_{int}$ denote the set of internal nodes in $\calT$. Let $v^*$ denote the root and $(u^*,v^*)$ denote its incoming edge. Then, 
\begin{align*}
& \sum_{v \in V_{int}} \sum_{u \in in(v)} H(\dsZ(u) \cup \dsD(u) |\dsW(u) \cup W_{new}(u)) \geq \\
& \sum_{v \in V_{int}} y_{(v,out(v))}+ \sum_{v \in V_{int}} H( \dsZ(v) \cup \dsD(v) | \dsW(v) \cup W_{new}(v)),
\end{align*}
where we have ignored the infinitesimal terms introduced due to Fano's inequality (for convenience of presentation).
Note that the RHS of the inequality above contains terms of the form $H( \dsZ(v) \cup \dsD(v) | \dsW(v) \cup W_{new}(v))$ for all nodes $v \in V_{int}$ (including $u^*$). 
On the other hand the LHS contains terms of a similar form for all nodes including the leaf nodes but excluding the node $u^*$.
Canceling the common terms, we obtain,
\begin{align*}
&\sum_{i=1}^{\ell} H(\dsZ(v_i)\cup \dsD(v_i) |W_{new}(v_i)) \geq \\
& \left( \sum_{v \in V_i} y_{(v,out(v))} \right)
 + H( Z\cup \dsD(u^*) | \dsW(u^*), W_{new}(u^*)),
\end{align*}
since $\dsW(v_i) = \phi$ for $i = 1,\dots, \ell$. We can therefore conclude that
\begin{align}
\label{eq:whymodlework}
\sum_{i=1}^{\ell} H(\dsZ(v_i), \dsD(v_i)) &\geq  \sum_{v \in V} y_{(v,out(v))}\\
\implies \sum_{i=1}^{\ell} H(\dsZ(v_i)) + \sum_{i=1}^{\ell} H(\dsD(v_i)) &\geq  \sum_{v \in V} y_{(v,out(v))}
\end{align}
Noting that $M \geq H(\dsZ(v_i))$ and $R^{\star} \geq H(\dsD(v_i))$ we have the required result. 
\end{proof}

\subsection{Proof of Claim \ref{claim:incomingedgelimits}}
\begin{proof}
We iteratively modify the problem instance $P(\calT,\alpha,\beta,L,N,K)$ to arrive at an instance where every node has in-degree at most two. Towards this end, we first identify a node $u$ with in-degree $\delta \geq 3$ such that no other node is topologically higher than it (such a node may not be unique).

We modify the instance $P$ by replacing $u$ with a directed in-tree where each node has in-degree exactly two. Specifically, arbitrarily number the nodes in $in(u)$ from $v_1', \dots, v'_\delta$. We replace the node $u$ with a directed in-tree $\calT_{u}$ with leaves  $v_1', \dots, v'_\delta$ and root $u$. $\calT_{u}$ has $\delta-2$ internal nodes numbered $u'_1, \dots, u'_{\delta-2}$ such that $in(u_{i}') = \{u'_{i-1}, v'_{i+1}\}$ where $u_0' = v_1'$ (see Fig. \ref{Fig:TreeModExmp}). Let us denote the new instance by $P_o = P_o(\calT_o,\alpha,\beta,L_o,N, K)$.
\begin{figure}[!t]
\begin{center}
%
%
%
\includegraphics[scale=0.5]{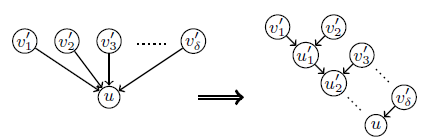}
\caption{{\small Tree modification example \label{Fig:TreeModExmp}}}
\end{center}
\end{figure}
We claim that $L_o \geq L$. 
To see this, suppose that $W^* \in W_{new}^{P}(u)$. We show that $W^* \in \cup_{u' \in \calT_u} W_{new}^{P_o}(u')$. This ensures that $L_o \geq L$.
To see this we note that
\begin{align*}
\dsZ^{P}(u) &= \dsZ^{P_o}(u)\\
\dsD^{P}(u) &= \dsD^{P_o}(u), \text{~and thus,}\\
\Delta^{P}(u,u) &= \Delta^{P_o}(u,u).
\end{align*}
Thus, if $W^* \in W_{new}^{P}(u)$, there exists an internal node $u'_i \in \calT_u$ with the smallest index $i \in \{1, \dots, \delta -2\}$ such that $W^* \in \Delta^{P_o}(u'_i,u'_i)$. Note that if $i > 1$, we have $W^* \in W_{new}^{P_o}(u'_i)$ since $W^* \notin \Delta^{P_o}(u'_{i-1},u'_{i-1})$ which in turn implies that $W^* \notin \dsW^{P_o}(u'_i)$. On the other hand if $i = 1$, then a similar argument holds since it is easy to see that $W^* \notin \dsW^{P_o}(u'_1)$.

Note that the modification in the instance $P$ can only affect nodes that are downstream of $u$. Now consider $u'$ such that $u \in in(u')$. It is evident that $\dsZ^{P_o}(u') = \dsZ^{P}(u')$ and $\dsD^{P_o}(u') =\dsD^{P}(u')$. Moreover $\dsW^{P_o}(u') = \cup_{v \in in(u')} \dsW^{P_o}(v) \cup W_{new}^{P_o}(v)$. Now for $v \neq u$, $\dsW^{P_o}(v) = \dsW^{P}(v)$ and $W_{new}^{P_o}(v) = W_{new}^{P}(v)$ as there are no changes in the corresponding subtrees. Moreover, as $\Delta^{P}(u,u) = \Delta^{P_o}(u,u)$, we have that $\dsW^{P_o}(u) \cup W_{new}^{P_o}(u) = \dsW^{P}(u) \cup W_{new}^{P}(u)$.
This implies that $\dsW^{P_o}(u') = \dsW^{P}(u')$. Thus, we can conclude that $W_{new}^{P_o}(u') = W_{new}^{P}(u')$. Applying an inductive argument we can conclude that the $W_{new}^{P_o}(u') = W_{new}^{P}(u')$ for all $u'$ such that $u \succ u'$.

The above process can iteratively be applied to every node in the instance that is of degree at least three. Thus, we have the required result.
\end{proof}

\subsection{Proof of Claim \ref{clm:Betahat_Equal_Beta}}
\begin{proof}
\label{appendix:proofBetahat_Equal_Beta}
We identify the set $\calU$ as the set of all nodes in $\calT$ such that the specified condition in the claim holds. Let $\calU^* \subset \calU$ denote the set of nodes that are highest in the topological ordering . We modify the instance in a way such that a node $u^* \in \calU^*$ can be removed from $\calU$, i.e., the specified condition no longer holds for it. Moreover, our modification procedure is such that a node $u \succ u^*$ cannot enter $\calU$ at the end of the procedure. 

We now discuss the modification procedure. In the discussion below, for a given node $u$, we can consider the instance obtained with tree $\calT_u$. We let $\beta_u$ denote the number of cache nodes in this instance.
Note that for $u^*$, the condition $\hat{\beta^*} < \min (\beta^*, K)$ holds. This implies that there is a set of cache leaves in $\calT_{u^*}$ denoted $\{v_{i_1}, \ldots, v_{i_m} \}$ such that $\dsZ(v_{i_1}) = \dots = \dsZ(v_{i_m}) = \{Z_j\}$. 
Let $\Lambda = \{u \in \calT_{u^*}: (v_{i_a}, v_{i_b}) \text{ meet at } u, \text{ for all distinct } v_{i_a}, v_{i_b} \in \{v_{i_1}, \ldots, v_{i_m} \}\}$. We identify $u_0 \in \Lambda$ such that no element of $\Lambda$ is topologically higher than $u_0$ (note that $u_0$ may not be unique) and let $v^*_{i_a}$ and $v^*_{i_b}$ be one pair of the corresponding nodes in $\{v_{i_1}, \ldots, v_{i_m} \}$ that meet at $u_0$. W.l.o.g we assume that $v^*_{i_b} \in \calT_{u_0(r)}$ and $v^*_{i_a} \in \calT_{u_0(l)}$.

We claim that $u_0=u^*$. Assume that this is not the case. Since $u_0 \in \calT_{u^*}$ we have $u_0 \succeq u^*$. Using this and the fact that $u_0 \notin \calU$ we have $|\cup_{v \in \calC_{u_0}}\dsZ(v)| = \min(|\calC_{u_0}|,K)$. Now, from $v^*_{i_a}, v^*_{i_b} \in \calC_{u_0}$ and that $\dsZ(v^*_{i_a})=\dsZ(v^*_{i_b})$ we conclude that $\min(|\calC_{u_0}|,K)=K$. Moreover, as $\cup_{u \in \calT_{u_0}}\dsZ(u) \subseteq \cup_{u \in \calT_{u^*}}\dsZ(u)$ we have $\hat{\beta} = K$ which contradicts $\hat{\beta} < \min(\beta,K)$. Therefore $u_0=u^*$.

We construct instance $P'$ (with lower bound $L'$) as follows. Choose a member of $\{ Z_1,\ldots, Z_K \} \setminus \{\dsZ(v'): v' \in \calC_{u^*} \}$ and denote it by $Z_k$. We set $\dsZ^{P'}(v^*_{i_b}) = \{Z_k\}$. Also, for any $u \in \calD_{u_0(r)}$ and $\dsD^{P}(u) = X_{d_1,\ldots,d_K}$ we set $\dsD^{P'}(u) = X_{d'_1, \dots, d'_K}$ such that $d'_j = d_k$ and $d'_k = d_j$ and $d'_i = d_i$ for $i \notin \{j,k\}$, i.e., we interchange the $j$-th and $k$-th labels and keep the other labels the same. With this modification, it can be seen that $\hat{\beta^*} = \min (\beta^*, K)$.

For nodes $u \succ u^*$, the change we applied to cache nodes in $\calC_{u^*}$ to get $P'$ is such that $\hat{\beta}_u$ continues to equal $\min (\beta_u, K)$ since $Z_k$ is chosen from $\{ Z_1,\ldots, Z_K \} \setminus \{\dsZ(v'): v' \in \calC_{u^*} \}$

We now show that $L' \geq L$. 
In particular, for $u \in \calT_{u_0(l)}$, we have $W_{new}^{P'}(u) = W_{new}^{P}(u)$, as there are no changes in the corresponding labels.
Also we claim that $W_{new}^{P'}(u)=W_{new}^{P}(u)$ for $u \in \calT_{u_0(r)}$. To see this, note that for $v \in \calD_{u_0(r)}$ and $v' \in \calC_{u_0(r)}$ we have $\Delta^{P'}(v',v) = \Delta^{P}(v',v)$ if $\dsZ(v') \notin \{Z_j, Z_k\}$. If $\dsZ^{P'}(v') = \{Z_k\}$ and $\dsD^{P'}(v) = X_{d'_1, \dots, d'_K}$ then,
\begin{align*}
\Delta^{P'}(v',v) &= Rec(\{Z_k\}, \{X_{d'_1, \dots, d'_K}\}) \\
&= \{W_{d'_k}\} = \{W_{d_j}\}\\
&= Rec(\{Z_j\}, \{X_{d_1, \dots, d_K}\}) \\
&= \Delta^{P}(v',v).
\end{align*}
Furthermore, note that there does not exist any $v' \in \calC_{u_0(r)}$ such that $\dsZ(v') = \{Z_j\}$ since we picked $u_0$ such that no element of $\Lambda$ is topologically higher than $u_0$. From eq. (\ref{eq:wnew_2}) and (\ref{eq:W-lab-union}), it is not hard to see that this in turn implies that $W_{new}^{P'}(u)=W_{new}^{P}(u)$ for $u \in \calT_{u_0(r)}$. 

It follows therefore that $\dsW^{P'}(u_0)  = \dsW^{P}(u_0)$ (from eq. (\ref{eq:W-lab-union})). Let us now consider the other nodes. As the changes are applied only to $\calT_{u_0(r)}$ so $label(u)$ changes only for nodes $u$ such that $u_0 \succ u$. Consider the subset of internal nodes $U=\{u_0, u_1, \ldots, u_t\}$ such that $(u_i, u_{i+1})$ is an edge, i.e., the set of internal nodes including $u_0$ and all nodes downstream of $u_0$ such that $u_t$ is the last internal node. W.l.o.g we assume that $u_{i-1}\in \calT_{u_{i}(l)}$ for $i \geq 1$. We now show that $\cup_{u \in U}W_{new}^{P}(u) \subseteq \cup_{u \in U}W_{new}^{P'}(u)$. Towards this end we have the following observations for $u \in U$.
\begin{align*}
\dsZ^{P'}(u) &= \dsZ^{P}(u) \cup \{Z_k\} \text{~(from the construction of $P'$)}\\
\Delta^{P'}(u,u) &= \cup_{v \in \calD_{u}} \Delta^{P'}(u,v).
\end{align*}
Now, for $v \notin \calD_{u_0(r)}$ we have $\dsD^{P'}(v) = \dsD^{P}(v)$ so that
\begin{align*}
\Delta^{P'}(u,v) &= Rec(\dsZ^{P'}(u),\dsD^{P'}(v))\\
&=Rec(\dsZ^{P'}(u),\dsD^{P}(v))\\
&\supseteq \Delta^{P}(u,v) (\text{~since~} \dsZ^{P'}(u) \supseteq \dsZ^{P}(u)). 
\end{align*}
Conversely for $v \in \calD_{u_0(r)}$ we have
\begin{align*}
Rec\left(\{Z_j, Z_k\}, \dsD^{P'}(v) \right) &= Rec\left(\{Z_j, Z_k\}, \dsD^{P}(v) \right),
\end{align*}
and
\begin{align*}
Rec\left(\{Z_i\}, \dsD^{P'}(v) \right) &= Rec\left(\{Z_i\}, \dsD^{P}(v) \right) \text{~~(for $Z_i \notin \{Z_j, Z_k \}$)}.
\end{align*}
Now, note that $\{Z_k, Z_j \} \subseteq \dsZ^{P'}(u)$ so that
\begin{align*}
\Delta^{P'}(u,v) &= Rec\left(\dsZ^{P'}(u), \dsD^{P'}(v) \right) \\
&= Rec\left(\dsZ^{P'}(u), \dsD^{P}(v) \right), \\
&\supseteq Rec\left(\dsZ^{P}(u), \dsD^{P}(v) \right)  = \Delta^{P}(u,v),
\end{align*}
since $\dsZ^{P'}(u) \supseteq \dsZ^{P}(u)$.
We can therefore conclude that
\begin{align*}
\Delta^{P}(u,u) &= \cup_{v \in \calD_{u}} \Delta^{P}(u,v) \subseteq \cup_{v \in \calD_{u}} \Delta^{P'}(u,v) = \Delta^{P'}(u,u).
\end{align*}
Now we consider a $W^* \in W_{new}^{P}(u_i)$ so that $W^* \in \Delta^{P}(u_i,u_i)$ which by above condition means that $W^* \in \Delta^{P'}(u_i,u_i)$.
Thus either $W^* \in W_{new}^{P'}(u_i)$ or $W^* \in \dsW^{P'}(u_i)$. In the latter case there exists a node $u_{i'}$ where $0 \leq i' < i$ such that $W^* \in W_{new}^{P'}(u_{i'})$ since $W^* \notin \dsW(u_0)$ and we have shown that $\dsW^{P'}(u_0)  =\dsW^{P}(u_0)$. Thus, we observe that
\begin{align*}
L' &= |\cup_{u \in U}W_{new}^{P'}(u)| + \sum_{u \in \calT', u \notin U} |W_{new}^{P'}(u)| , \\
&\geq |\cup_{u \in U}W_{new}^{P}(u)| + \sum_{u \in \calT, u \notin U} |W_{new}^{P}(u)| , \\
&= L,
\end{align*}
where the second inequality holds since $\sum_{u \in \calT', u \notin U} |W_{new}^{P'}(u)| = \sum_{u \in \calT, u \notin U} |W_{new}^{P}(u)|$ and $|\cup_{u \in U}W_{new}^{P'}(u)| \geq |\cup_{u \in U}W_{new}^{P}(u)|$.

As discussed before, the modification procedure is such that at the end of the operation $u^* \notin \calU$. Moreover nodes $u \succ u^*$ are not in $\calU$ either.
For each node $u \in \calU$ let $d(u)$ denote the number of edges in path connecting $u$ to the root node. Our modification procedure is such that $d^* = \max_{u \in \calU} d(u)$ is guaranteed to decrease over the course of the iterations. Indeed, if $|\calU^*| =1$, then at the end of the iteration $d^*$ will definitely decrease. If $|\calU^*| > 1$, then $d^*$ will definitely decrease after the modification procedure is applied to all the nodes in $\calU^*$. Thus, the sequence of iterations is guaranteed to terminate. This observation concludes the proof.  

\end{proof}

\subsection{Proof of Lemma \ref{lemma:Increase_N_Increase_L}}
\begin{proof}

Given the conditions of the theorem, from Corollary \ref{corollay:sat_rem} we can conclude that there exists an index $i^* \in \{1, \dots, \alpha\}$ such that $\sum_{v' \in \calC} \psi(v_{i^*},v') < \min(\beta, K)$. We set $i^*$ to be the smallest such index. Let $\Pi^1(v_{i^*}) = \{v' \in \calC: \psi(v_{i^*}, v') = 1\}$ and $\Pi^0(v_{i^*}) = \{v' \in \calC: \psi(v_{i^*}, v') = 0, \dsZ(v') \nsubseteq \cup_{v \in \Pi^1(v_{i^*})} \dsZ(v) \}$. Note that  $\Pi^0(v_{i^*})$ is non-empty since $|\cup_{v' \in \calC} \dsZ(v')| = \min(\beta,K)$ and $\sum_{v' \in \calC} \psi(v_{i^*},v') < \min(\beta, K)$.

Next, we determine the set of nodes where $v_{i^*}$ and the nodes in $\Pi^0(v_{i^*})$ meet, i.e., we define $\Lambda^{0}(v_{i^*}) = \{u \in \calT: \exists v' \in \Pi^0(v_{i^*}) \text{~such that~} v_{i^*} \text{~and~} v' \text{~meet at~} u.\}$. Note that there is a topological ordering on the nodes in $\Lambda^{0}(v_{i^*})$. Pick the node $u^* \in \Lambda^{0}(v_{i^*})$ such that no element of $\Lambda^{0}(v_{i^*})$ is topologically higher than $u^*$ ($u^*$ is in the path from $v_{i^*}$ to the root node). Let the corresponding node in $\Pi^0(v_{i^*})$ be denoted by $v_{j^*}$ where $j^* \in \{\alpha + 1, \dots, \alpha+\beta\}$. Note that $v_{j^*}$ might not be unique.

Suppose that $\dsZ(v_{j^*}) = \{Z_k\}$ and that $\dsD(v_{i^*}) = X_{d_1, \dots, d_K}$. We modify the instance $P$ as follows. Set $d_k = N+1$ (i.e., the index of the $N+1$ file). Thus, the only change is in $\dsD(v_{i^*})$. Let us denote the new instance by $P'=P(\calT',\alpha,\beta,L',N+1,K)$.

We now analyze the value of $L'$.
W.l.o.g. we assume that $v_{i^*} \in \calT'_{u^*(l)}$ and $v_{j^*} \in \calT'_{u^*(r)}$. 
Note that $W_{new}^{P'}(u) = W_{new}^{P}(u)$ for $u \in \calT'_{u^*(r)}$ as the subtree $\calT'_{u^*(r)}$ is identical to $\calT_{u^*(r)}$. We also have
$$ W_{new}^{P'}(u) = W_{new}^{P}(u) \text{~for~} u \in \calT'_{u^*(l)}.$$
To see this suppose that this is not true. This implies that the file $W_{N+1}$ is recovered at some node in $\calT'_{u^*(l)}$, i.e., there exists $v' \in \calC$ such that $v' \in \calT'_{u^*(l)}$, $\dsZ(v') = \{Z_k\}$, and that $v'$ and $v_{i^*}$ meet at some $u \succ u^*$. From $v_{j^*} \in \Pi^0(v_{i^*})$ we can conclude that $\{Z_k\} \nsubseteq \cup_{v \in \Pi^1(v_{i^*})}$ and $v' \in \Pi^0(v_{i^*})$ (as $\dsZ(v')=\{Z_k\}$).
However this is a contradiction, since this implies the existence of node $u$ that is topologically higher than $u^*$ in the set $\Lambda^{0}(v_{i^*})$. It follows from eq. (\ref{eq:W-lab-union}) that $\dsW^{P'}(u^*) = \dsW^{P}(u^*)$.

Next, we claim that $W_{new}^{P'}(u^*) =W_{new}^{P}(u^*) \cup \{W_{N+1}\}$. To see this consider the following series of arguments. Let the singleton subset $\Delta^{P}(v_{i^*}, v_{j^*}) = \{W^*\}$. 
Note that $\psi^{P}(v_{i^*},v_{j^*}) = 0$. This implies that there exist $v \in \calD_{u^*}$ and $v' \in  \calC_{u^*}$ such that $v$ and $v'$ meet above $u^*$ and recover the file $W^*$ where $(v,v') \neq (v_{i^*}, v_{j^*})$. 
Thus, as $\dsZ^{P'}(u^*) = \dsZ^{P}(u^*)$, we can conclude that
\begin{align*}
\Delta^{P'}(u^*, u^*) &= Rec(\dsZ^{P'}(u^*),  \dsD^{P'}(u^*)) \\
&= Rec(\dsZ^{P}(u^*),  \dsD^{P'}(u^*)) \\
&=  \Delta^{P}(u^*, u^*) \cup \{W_{N+1}\}.
\end{align*}
Furthermore, we have
\begin{align*}
&W_{new}^{P'}(u^*) = \Delta^{P'}(u^*,u^*) \setminus \dsW^{P'}(u^*)\\
&= \Delta^{P}(u^*,u^*) \cup \{W_{N+1}\} \setminus \dsW^{P}(u^*)\\
&= W_{new}^{P}(u^*) \cup \{W_{N+1}\}, \text{~(since~} W_{N+1} \notin \dsW^{P}(u^*)).
\end{align*}
For $u$ such that $u^* \succ u$ we inductively argue that $W_{new}^{P'}(u) = W_{new}^{P}(u)$. To see this suppose that $u^* = u_r$. It is evident that $\Delta_{rl}^{P'}(u) = \Delta_{rl}^{P}(u)$. Next,
$\Delta_{lr}^{P'}(u) = \Delta_{lr}^{P}(u)$ since $Z_k \notin \dsZ(u_l) \setminus \dsZ(u_r)$. Thus,
\begin{align*}
&W_{new}^{P'}(u) = \Delta_{rl}^{P'}(u) \cup \Delta_{lr}^{P'}(u) \setminus \dsW^{P'}(u)\\
&= \Delta_{rl}^{P}(u) \cup \Delta_{lr}^{P}(u) \setminus \dsW^{P'}(u)\\
&= \Delta_{rl}^{P}(u) \cup \Delta_{lr}^{P}(u) \setminus \dsW^{P}(u) \cup \{W_{N+1}\}\\
&= \Delta_{rl}^{P}(u) \cup \Delta_{lr}^{P}(u) \setminus \dsW^{P}(u)  \text{~(since $W_{N+1} \notin \Delta_{rl}^{P}(u) \cup \Delta_{lr}^{P}(u)$)}\\
&= W_{new}^{P}(u).
\end{align*}

Next, we note that $\dsW(u) = \dsW(u_r) \cup W_{new}(u_r) \cup \dsW(u_l) \cup W_{new}(u_l)$. It is evident that $\dsW^{P'}(u_l)=\dsW^{P}(u_l)$ and $W_{new}^{P'}(u_l)=W_{new}^{P}(u_l)$. Next, $\dsW^{P'}(u_r) = \dsW^{P'}(u^*) = \dsW^{P}(u^*)$ (from above) and $W_{new}^{P'}(u^*) = W_{new}^{P}(u^*) \cup \{W_{N+1}\}$, so that $\dsW^{P'}(u) = \dsW^{P}(u) \cup \{W_{N+1}\}$.

As the induction hypothesis we assume that for any node $u$ downstream of $u^*$, we have $W_{new}^{P'}(u) = W_{new}^{P}(u)$ and $\dsW^{P'}(u) = \dsW^{P}(u) \cup \{W_{N+1}\}$. Consider a node $u'$ such that $u'_r = u$.
As before we have $\dsW^{P'}(u'_l)=\dsW^{P}(u'_l)$, $W_{new}^{P'}(u'_l)=W_{new}^{P}(u'_l)$. Moreover, we have $\dsW^{P'}(u'_r) = \dsW^{P}(u'_r) \cup \{W_{N+1}\}$ and $W_{new}^{P'}(u'_r) = W_{new}^{P}(u'_r)$, by the induction hypothesis, so that $\dsW^{P'}(u') = \dsW^{P}(u') \cup \{W_{N+1}\}$.

Next, we argue similarly as above that $\Delta_{rl}^{P'}(u') = \Delta_{rl}^{P}(u')$ and $\Delta_{lr}^{P'}(u') = \Delta_{lr}^{P}(u')$ and the sequence of equations above can be used to conclude to that $W_{new}^{P'}(u')=W_{new}^{P}(u')$.

We conclude that $L' = L+1$.
\end{proof}

\subsection{Proof of Claim \ref{clm:Gamma_right_subset_Gamma_left}}
\label{appendix:proofGamma_right_subset_Gamma_left}
\begin{proof}

W.l.o.g we assume that $|\Gamma_l| \geq |\Gamma_r|$ for all $u \in \calT$. We identify the set $\calU$ as the set of nodes in $\calT$ such that $\Gamma_r \nsubseteq \Gamma_l$. Let $\calU^* \subset \calU$ denote the set of nodes in $\calU$ that are highest in the topological ordering.

Consider a node $u^* \in \calU^*$. Note that since $|\Gamma_l| \geq |\Gamma_r|$, there exists an injective mapping $\phi:\Gamma_r \setminus \Gamma_l \rightarrow \Gamma_l \setminus \Gamma_r$. Let $\dsZ(u^*_r) = \{Z_{i_1}, \dots, Z_{i_m}\}$.
We construct the instance $P'$ as follows. For each $v \in \calD_{u^*_r}$ suppose $\dsD(v) = \{X_{d_1, \dots, d_K}\}$. For $j = 1, \dots, m$, if $d_{i_j} \in \Gamma_r \setminus \Gamma_l$, we replace it by $\phi(d_{i_j})$; otherwise, we leave it unchanged. In other words, we modify the delivery phase signals  so that the files that are recovered in $\calT_{u^*(r)}$ are a subset of those recovered in $\calT_{u^*(l)}$.


As our change amounts to a simple relabeling of the sources, for $u \in \calT_{u^*(r)}$ we have $|W_{new}^{P'}(u)| = |W_{new}^{P}(u)|$. For any $u \succ u^*$ we have $\Gamma^P_r(u) \subseteq \Gamma^P_l(u)$. Similarly, we can show that $\Gamma^{P'}_r(u) \subseteq \Gamma^{P'}_l(u)$. We note that $\Gamma^{P'}$ and $\Gamma^{P}$ only differ in files like $W_d$ where $d$ is in domain of $\phi(\cdot)$, i.e., if $W_d \in \Gamma^P$ then $W_{\phi(d)} \in \Gamma^{P'}$.
If there exist a file $W_d\in \Gamma^P_r(u)$ with $d$ in domain of $\phi(\cdot)$ then $W_{\phi(d)} \in \Gamma^{P'}_r(u)$ and from $\Gamma^P_r(u) \subseteq \Gamma^P_l(u)$ we have $W_{\phi(d)} \in \Gamma^{P'}_l(u)$. Thus, we have $\Gamma^{P'}_r(u) \subseteq \Gamma^{P'}_l(u)$. This indicates that after applying this change, the property of $\Gamma_r \subseteq \Gamma_l$ still holds in $P'$ for all nodes $u$ that are upstream of $u^*$.
Furthermore, the relabeling of the sources only affects $u \in \calT'$ such that $u^* \succ u$. Note that $\dsW^{P'}(u^*) \subset \dsW^{P}(u^*)$ (the inclusion is strict since at least one source in $\Gamma_r \setminus \Gamma_l$ is mapped to $\Gamma_l \setminus \Gamma_r$) since we have $\Gamma^{P'}_r \subseteq  \Gamma^{P'}_l$ and $\Gamma^{P'}_l =  \Gamma^{P}_l$.

Now, we note that
\begin{align*}
\Delta_{rl}^{P'}(u^*) &= \Delta_{rl}^{P}(u^*), \text{~and}\\
\Delta_{lr}^{P'}(u^*) &= \Delta_{lr}^{P}(u^*),
\end{align*}
where the first equality holds since $\dsZ^{P}(u^*_r) =\dsZ^{P'}(u^*_r)$, $\dsZ^{P}(u^*_l) =\dsZ^{P'}(u^*_l)$ and $\dsD^{P}(u^*_l) =\dsD^{P'}(u^*_l)$. The second equality holds since our modification to the delivery phase signals in $\calT_{u^*(r)}$ does not affect files that are recovered from $\dsZ^{P}(u^*_l) \setminus \dsZ^{P}(u^*_r)$. It follows therefore that $|W_{new}^{P'}(u^*)| \geq |W_{new}^{P}(u^*)|$.

We make an inductive argument for nodes $u$ that are downstream of $u^*$; w.l.o.g. we assume that $u^* \in \calT_{u(r)}$. Specifically, our inductive hypothesis is that for a node $u$ that is downstream of $u^*$, we have $\dsW^{P'}(u) \subseteq \dsW^{P}(u)$, $\Delta_{rl}^{P'}(u) = \Delta_{rl}^{P}(u)$ and $\Delta_{lr}^{P'}(u) = \Delta_{lr}^{P}(u)$.

Now consider a node $u'$ downstream of $u$ such that $u'_r = u$.  We have, $\dsW(u') = \dsW(u'_l) \cup W_{new}(u'_l) \cup \dsW(u) \cup W_{new}(u)$. Note that we can express $ \dsW(u) \cup W_{new}(u) = \dsW(u) \cup \Delta_{rl}(u) \cup \Delta_{lr}(u)$ . It is evident that $\dsW^{P'}(u'_l) = \dsW^{P}(u'_l)$ and $W_{new}^{P'}(u'_l) = W_{new}^{P}(u'_l)$. Moreover, by the induction hypothesis, $\dsW^{P'}(u) \subseteq \dsW^{P}(u)$ and $\Delta_{rl}^{P'}(u) \cup \Delta_{lr}^{P'}(u) =  \Delta_{rl}^{P}(u) \cup \Delta_{lr}^{P}(u)$. Thus, the induction step is proved.


We have shown that after applying the changes for $u^*$, the condition $\Gamma_r \nsubseteq \Gamma_l$ will not hold for $u \succeq u^*$.
For each node $u \in \calU$ let $d(u)$ denote the number of edges in path connecting $u$ to the root node. Our modification procedure is such that $d^* = \max_{u \in \calU} d(u)$ is guaranteed to decrease over the course of the iterations. Indeed, if $|\calU^*| =1$, then at the end of the iteration $d^*$ will definitely decrease. If $|\calU^*| > 1$, then $d^*$ will definitely decrease after the modification procedure is applied to all the nodes in $\calU^*$. Thus, the sequence of iterations is guaranteed to terminate. This observation concludes the proof.  
\end{proof}

\begin{claim}
\label{clm:app_rho_eq_rho_tild}
Under condition of $\hat{\beta}_l = \min(\beta_l,K)$ and $\hat{\beta}_r = \min(\beta_r,K)$ we have $\min(\hat{\beta}_l,K-\hat{\beta}_r) = [\min(\beta_l,K-\beta_r)]^+$ and $\min(\hat{\beta}_r,K-\hat{\beta}_l) = [\min(\beta_r,K-\beta_l)]^+$.
\end{claim}
\begin{proof}
First, we consider the case where $\beta_l+\beta_r \leq K$ so $\beta_l \leq K-\beta_r$ and $[\min(\beta_l,K-\beta_r)]^+ = \beta_l$. By assumption, $\beta_l+\beta_r \leq K$ implies $\hat{\beta}_l+\hat{\beta}_r \leq K$ thus $\min(\hat{\beta}_l,K-\hat{\beta}_r) = \hat{\beta}_l = \beta_l$. We now consider the $\beta_l+\beta_r \geq K$ case which in turns leads to $\hat{\beta}_l +\hat{\beta}_r \geq K$. Therefore, $$\min(\hat{\beta}_l, K-\hat{\beta}_r) = K-\hat{\beta}_r=K-\min(K,\beta_r) = \max(0,K-\beta_r) = [K-\beta_r]^+ = [\min(\beta_l,K-\beta_r)]^+.$$
The same argument will show that $\min(\hat{\beta}_r,K-\hat{\beta}_l) = [\min(\beta_r,K-\beta_l)]^+$.
\end{proof}

\begin{claim}
\label{clm:ab_eq_albl_plus_arbr_plus}
Consider the integers $\alpha,\alpha_l,\alpha_r, \beta,\beta_l,\beta_r,K$ so that $\alpha= \alpha_l+\alpha_r$ and $\beta=\beta_l+\beta_r$. Then
\beqa
\alpha \min(\beta,K) = \alpha_l \min(\beta_l,K) + \alpha_r \min(\beta_r,K) + \alpha_l [\min(\beta_r,K-\beta_l)]^++\alpha_r [\min(\beta_l,K-\beta_r)]^+.
\eeqa
\end{claim}
\begin{proof}
First, we consider the case where $\beta \leq K$ thus $\beta_l \leq K-\beta_r$ and $\beta_r \leq K-\beta_l$. Then, the above relation reduces to $\alpha\beta = \alpha_l \beta_l + \alpha_r \beta_r + \alpha_l\beta_r+\alpha_r \beta_l$ which is true. For the case $\beta \geq K$, the relation reduces to $\alpha K = \alpha_l \left( \min(\beta_l,K) +[K-\beta_l]^+\right) + \alpha_r \left(\min(\beta_r,K)+[K-\beta_r]^+ \right) $. This equation holds since $\min(\beta_l,K) =K-[K-\beta_l]^+$, and same thing for $\beta_r$.
\end{proof}

\end{document}